\documentclass{article}

\usepackage{arxiv}
\usepackage{graphicx}
\usepackage[T1]{fontenc}
\usepackage[utf8]{inputenc}
\usepackage[english]{babel}
\usepackage{lmodern}
\usepackage{times}
\usepackage{hyperref}       
\usepackage{amsmath}
\usepackage{amsthm}
\usepackage{amsfonts}       

\usepackage[ruled,vlined]{algorithm2e}
\SetKwFor{round}{\hspace*{-0.78ex}}{}{end}
\SetKwFunction{randomColor}{randomColor}
\SetKw{KwRet}{return}
\SetKw{Break}{break}
\SetKwInOut{Input}{input}

\RequirePackage{mleftright} 
\RequirePackage{xparse}
\RequirePackage{mathtools}
\RequirePackage{todonotes}

\newcommand{\suchthat}{\;\ifnum\currentgrouptype=16 \middle\fi|\;}

\NewDocumentCommand{\abs}{som}{%
  \IfNoValueTF{#2}{
    \IfBooleanTF{#1}{%
      \left\lvert\mskip0.3\thinmuskip#3\mskip0.3\thinmuskip\right\rvert%
    }{%
      \lvert\mskip0.3\thinmuskip#3\mskip0.3\thinmuskip\rvert%
    }%
  }{%
    \mathopen{#2\lvert}\mskip0.3\thinmuskip#3\mskip0.3\thinmuskip\mathclose{#2\rvert}
  }%
}%

\newtheorem{theorem}{Theorem}

\newtheorem{definition}[theorem]{Definition}
\newtheorem{lemma}[theorem]{Lemma}
\newtheorem{corollary}[theorem]{Corollary}

\title{Fixed Points and 2-Cycles of Synchronous
  Dynamic Coloring Processes on Trees}
\author{%
  \href{https://orcid.org/0000-0001-9964-8816}{%
    \includegraphics[height=0.8em]{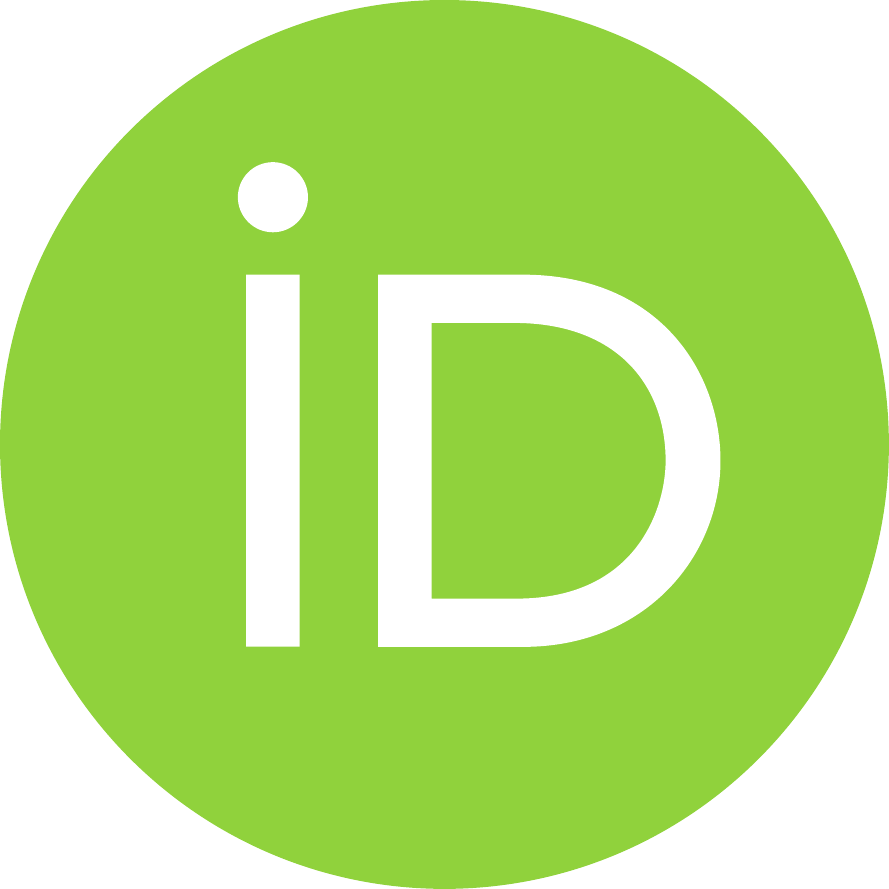}%
    \hspace{1mm}%
    Volker Turau%
  } \\
  Institute of Telematics \\
  Hamburg University of Technology \\
  21073 Hamburg, Germany \\
  \texttt{turau@tuhh.de} }

\begin{document}

\maketitle

\begin{abstract}
  This paper considers synchronous discrete-time dynamical systems on
  graphs based on the threshold model. It is well known that after a
  finite number of rounds these systems either reach a fixed point or
  enter a 2-cycle. The problem of finding the fixed points for this
  type of dynamical system is in general both NP-hard and
  \#P-complete. In this paper we give a surprisingly simple
  graph-theoretic characterization of fixed points and 2-cycles for
  the class of finite trees. Thus, the class of trees is the first
  nontrivial graph class for which a complete characterization of
  fixed points exists. This characterization enables us to provide
  bounds for the total number of fixed points and pure 2-cycles. It
  also leads to an output-sensitive algorithm to efficiently generate
  these states.
\end{abstract}

\section{Introduction}
\label{sec:intro}

Synchronous discrete-time dynamical systems for information spreading
received a lot of attention in recent years. Often the following model
is used: Let $G$ be a graph with an initial configuration, where each
node is either black or white. In discrete-time rounds, all nodes
simultaneously update their color based on a predefined local rule.
The rule is local in the sense that the color associated with a node
in round $t$ is determined by the colors of the neighboring nodes in
round $t-1$. The main focus of the research so far has been on the
stabilization time of this process \cite{Zehmakan:2019} and the
dominance problem, e.g., how many nodes must initially be black so
that eventually all nodes are black \cite{Peleg:2002}. These questions
have been considered for various classes of graphs. These
discrete-time dynamical systems are often based on the {\em threshold
  model}. In a simple version of this model a node becomes black
if at least a fraction of $\alpha$ of its neighbors are black and
white otherwise, $\alpha\in (0,1)$ is a parameter of the model. In
more elaborate versions edges have weights and the local rules are
based on the weighted fraction of neighbors. The main property of
these dynamical systems is that assuming symmetric weights, the system
has period $1$ or $2$ \cite{Goles:1980,Poljak:1983}. This means that
such a system eventually reaches a stable configuration or it toggles
between two configurations. Fogelman et al.\ proved that the
stabilization time is in $O(n^2)$ \cite{Fogelman:1983}. Frischknecht
et al.\ showed that this bound is tight, up to some poly-logarithmic
factor \cite{Frischknecht:2013}.

In this paper we analyze a different aspect of discrete-time dynamical
systems: The number and structure of fixed points and 2-cycles. This
research is motivated by applications of so called Boolean networks
(BN) \cite{Kauffman:1993}, i.e., discrete-time dynamical systems,
where each node (e.g., gene) takes either 0 (passive) or 1 (active)
and the states of nodes change synchronously according to regulation
rules given as Boolean functions. An example for a regulation rule is
the majority rule, i.e., $\alpha=0.5$. Since the problem of finding
the fixed points of a BN is in general both NP-hard and \#P-complete
\cite{Akutsu:1998} (see Sec.~\ref{sec:state-art}), it is interesting
to find graph classes, for which the number of fixed points can be
determined efficiently. We regard our work as a step in this
direction. Interest in the set of fixed points of BNs was also sparked
by a result of Milano and Roli \cite{Milano:2000}. They use BNs to
solve the satisfiability problem (SAT) by defining a mapping between a
SAT instance and a BN and prove that BN fixed points correspond to SAT
solutions.

This paper provides a characterization of the set of stable
configurations (a.k.a.\ fixed points) and the set of states of period
$2$ (a.k.a.\ 2-cycles) for a given finite tree based on its edge set.
We do this for two versions of the threshold model: minority and
majority process. While the stabilization times for the majority and
minority process can differ considerably for a given graph (see
Fig.~\ref{fig:minority}), the sets of stable configurations of a tree
turn out to be closely related for both process types. Our main
contributions  are as follows:
\begin{enumerate}
\item We identify a subset $E_{fix}(T)$ of the power set of the edge
  set of a tree $T$ and show that the elements of $E_{fix}(T)$
  correspond one-to-one with the fixed points of $T$. $E_{fix}(T)$ is
  defined by a set of simple linear inequalities over the node
  degrees. The fixed point corresponding to an element of $E_{fix}(T)$
  can be defined in simple terms. $E_{fix}(T)$ has the hereditary
  property, i.e., if $X\in E_{fix}(T)$ then all subset of $X$ are also
  elements of $E_{fix}(T)$. This property allows to define a simple
  output-sensitive algorithm $\mathcal{A_{M}}$ to explicitly generate
  all fixed points. This allows to prove upper bounds for the number
  of fixed points. We also show that elements of $E_{fix}(T)$
  correspond to solutions of a system of linear diophantine
  inequalities.
\item We characterize the configurations of period $2$, where each
  node changes its color in every round (a.k.a.\ pure configurations).
  As above we identify a subset $E_{pure}(T)$ of the power set of the edge set
  of $T$ such that the elements of $E_{pure}(T)$ correspond one-to-one
  with the pure configurations of $T$. As above the definition of
  $E_{pure}(T)$ is based on simple linear inequalities and it has the
  hereditary property. The 2-cycle corresponding to an element of
  $E_{pure}(T)$ is also defined in simple terms. Again this allows to
  define a simple algorithm enumerating all 2-cycles and to prove
  upper bounds for their number. Interestingly, $E_{pure}(T)$ is a
  subset of $E_{fix}(T)$.
\item Finally we look at general configurations with period $2$. We
  show that for each configuration $c$ of this type each tree
  decomposes into subtrees, such that $c$ induces either a fixed point
  or a pure configuration on each subtree. The subtrees allow to
  define a hyper structure of a tree, called the {\em block tree}. As
  in previous cases we identify a subset $E_{block}(T)$ of the
  power set of the edge set of a tree $T$ and show that the elements
  of $E_{block}(T)$ correspond one-to-one with the block trees of $T$.
  $E_{block}(T)$ is a subset of $E_{fix}(T)$. Since a tree can have
  several pure colorings, a block tree does not uniquely define a
  coloring. We define a subclass of 2-cycles called canonical
  colorings and prove that there is a direct correspondence between
  $E_{block}(T)$ and canonical colorings. The characterization of
  $E_{block}(T)$ is not as simple as in the above cases, since
  $E_{block}(T)$ does not have the hereditary property.
\end{enumerate}
All results are obtained for the minority and the majority model.

\section{State of the Art}
\label{sec:state-art}
Most research on discrete-time dynamical systems on graphs consecrates
oneself to bounds of the stabilization time. Good overviews for the
majority resp.\ the minority process can be found in
\cite{Zehmakan:2019} resp.\ \cite{Papp:2019}. Rouquier et al.\ study
the minority process in the asynchronous model, i.e., not all nodes
update their color concurrently \cite{Rouquier:2011}. They show that
the stabilization time strongly depends on the topology and observe
that the case of trees is non-trivial.

The analysis of fixed points of the majority or minority process
received only some attention. Kr{\'a}lovi{\v{c}} determined the number
of fixed points of a complete binary tree for the majority process
\cite{Rastislav:2001}. Agur et al.\ did the same for ring topologies
\cite{Agur:1988}. In both cases the number of fixed points is an
exponentially small fraction of all configurations.

A related concept are Boolean networks, which have been extensively
used as mathematical models of genetic regulatory networks. The number
of fixed points of such a network is a key feature of its dynamical
behavior. Boolean networks have been extensively used as models for
the dynamics of gene regulatory networks. A gene is modeled by binary
values, 0 or 1, indicating two transcriptional states, either active
or inactive, respectively. Each network node operates by the same
nonlinear majority rule, i.e., majority processes are a particular
type of BN \cite{Veliz:2012}. The set of fixed points is an important
feature of the dynamical behavior of such networks
\cite{Aracena:2008}. The number of fixed points is a measure for the
general memory storage capacity of a system. Many fixed points imply
that a system can store a large amount of information, or, in
biological terms, has a large phenotypic repertoire \cite{Agur:1991}.
However, the problem of finding the fixed points of a Boolean network
is in general both NP-hard and \#P-complete \cite{Akutsu:1998}. There
are only a few theoretical results to efficiently determine this set
\cite{Irons:2006}. Aracena determined the maximum number of fixed
points in a particular class of BN called regulatory Boolean networks
\cite{Aracena:2008}.

Concepts related to fixed points of the minority resp.\ majority
process have been analyzed. A partition $(S,\bar{S})$ of the nodes of
a graph is called a global defensive $0$-alliance if
$\abs{N_S(v)}\ge \abs{N_{\bar{S}}(v)}$ for each node $v$
\cite{Fernau:2014}. Thus, a fixed point of the minority or majority
process induces a $0$-alliance, but the converse does not hold. The
difference is that no condition is placed on $v$. $0$-alliances are
also called monopolies \cite{Mishra:2006}.

Mishra and Rao show that, for trees, a minimum monopoly can be
computed in linear time \cite{Mishra:2006}. A related concept is that
of $2$-community structure. Bazgan et al.\ prove that each tree has a
connected $2$-community structure and it can be found in linear time
\cite{Bazgan:2010}.

\section{Synchronous Discrete-Time Dynamical Systems}
Let $G(V,E)$ be a finite, undirected graph. A coloring $c$ assigns to
each node of $G$ a value of $\{0,1\}$ with no further constraints on
$c$. Denote by $\mathcal{C}(G)$ the set of all colorings of $G$, i.e.,
$\abs{\mathcal{C}(G)}=2^{\abs{V}}$. A transition process $\mathcal{M}$
describes the transition of one coloring to another, i.e., it is a
mapping $\mathcal{M}: \mathcal{C}(G) \longrightarrow \mathcal{C}(G)$.
Given an initial coloring $c$, a transition process produces a sequence of
colorings $c, \mathcal{M}(c), \mathcal{M}(\mathcal{M}(c)),\ldots$. We
consider two transition processes: {\em Minority} and {\em Majority}
process and denote the corresponding mappings by $\mathcal{MIN}$ and
$\mathcal{MAJ}$. They are local mappings in the sense that the new
color of a node is based on the current colors of its neighbors. To
determine $\mathcal{M}(c)$ the local mapping is executed concurrently
by all nodes. The transition from $c$ to $\mathcal{M}(c)$ is called a
{\em round}. In the minority (resp.\ majority) process each node
adopts the minority (resp.\ majority) color among all neighbors. In
case of a tie the color remains unchanged. Formally, the minority
process is defined for a node $v$ as follows:
\[\mathcal{MIN}(c)(v) = \begin{cases}
    c(v) & \text{if } \abs{N^{c(v)}(v)} \le  \abs{N^{1-c(v)}(v)}\\
    1-c(v)  & \text{if } \abs{N^{c(v)}(v)} > \abs{N^{1-c(v)}(v)} 
  \end{cases}\] $N^i(v)$ denotes the set of $v$'s neighbors with color
$i$ ($i=0,1$). The definition of $\mathcal{MAJ}$ is similar, only the
binary operators $\le$ and $>$ are reversed.
Sometimes a result holds for both processes. To simplify matters in
these cases we use the symbol $\mathcal{{M}}$ as a placeholder for
$\mathcal{{MIN}}$ and $\mathcal{{MAJ}}$. Fig.~\ref{fig:minority}
depicts a sequence of colorings for $\mathcal{{MIN}}$.

\begin{figure}[h]
  \hfill
  \includegraphics[scale=0.7]{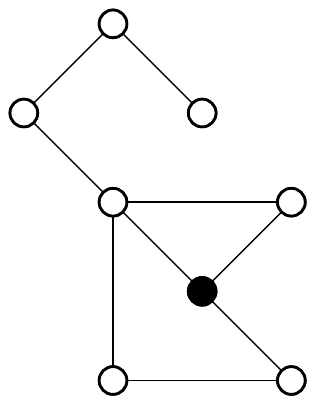}
  \hfill
  \includegraphics[scale=0.7]{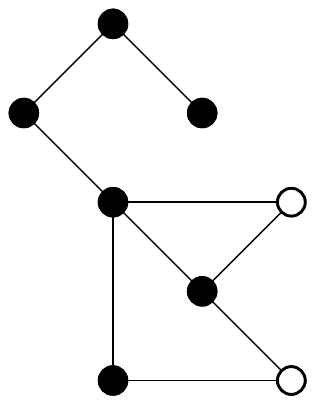}
  \hfill
  \includegraphics[scale=0.7]{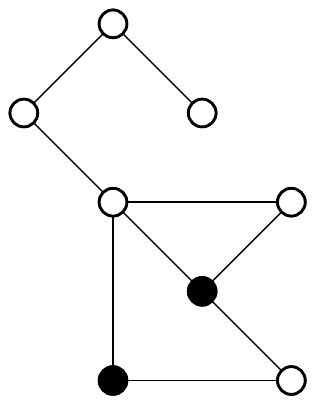}
  \hfill
  \includegraphics[scale=0.7]{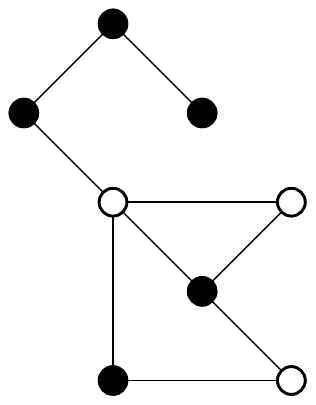}
  \hfill
  \includegraphics[scale=0.7]{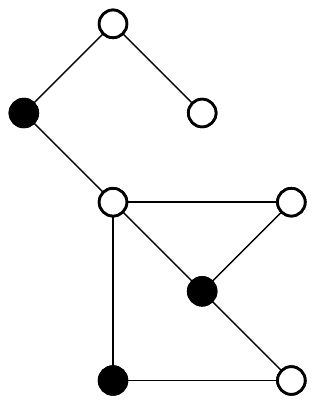}
  \hfill
  \includegraphics[scale=0.7]{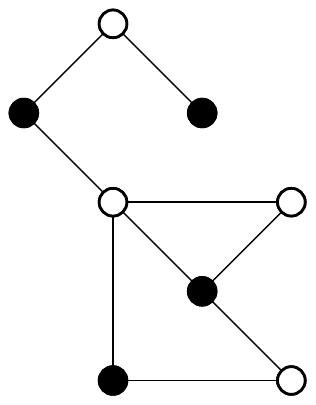}
  \hfill\null
  \caption{For the initial coloring on the left $\mathcal{{MIN}}$
    reaches after five rounds the coloring shown on the right.
    $\mathcal{{MAJ}}$ reaches for the same initial coloring after one
    round a monochromatic coloring.}\label{fig:minority}
\end{figure}

In this paper we are interested in colorings with specific properties.
Let $c\in \mathcal{C}(G)$. If $\mathcal{M}(c)=c$ then $c$ is called a
{\em fixed point}. It is called a {\em 2-cycle} if
$\mathcal{M}(c)\not=c$ and $\mathcal{M}(\mathcal{M}(c))=c$. A 2-cycle
is called {\em pure} if $\mathcal{M}(c)(v) \not =c(v)$ for each node
$v$ of $G$. $c$ is called {\em monochromatic} if all nodes have the
same color, i.e., $c(v)=c(w)$ for all $v,w\in V$. $c$ is called {\em
  independent} if the color of each node is different from the colors
of all its neighbors. Clearly, a monochromatic (resp.\ independent)
coloring is a fixed point for the majority (resp.\ minority) process.
An edge $(v,w)$ is called {\em monochromatic} for $c$ if $c(v)=c(w)$
otherwise it is called {\em multi-chromatic}.

For a mapping $\mathcal{M}$ denote by $\mathcal{F}_{\mathcal{M}}(G)$,
$\mathcal{C}^2_{\mathcal{M}}(G)$, and $\mathcal{P}_{\mathcal{M}}(G)$,
the set of all $c\in \mathcal{C}(G)$ that constitute a fixed point, a
2-cycle, or a pure coloring for $\mathcal{M}$. If $c$ belongs to one
of these sets the complementary coloring and $\mathcal{M}(c)$ also
belong to this set. To cope with this fact we also define the sets
$\mathcal{F}_{\mathcal{M}}(G)^+$, $\mathcal{C}^2_{\mathcal{M}}(G)^+$,
and $\mathcal{P}_{\mathcal{M}}(G)^+$ as the subsets of those colorings
of the corresponding sets which assign to a globally distinguished
node $v^\ast$ color $0$. Hence, if $c\in \mathcal{F}_{\mathcal{M}}(G)$
then either $c$ or the complement of $c$ is in
$\mathcal{F}_{\mathcal{M}}(G)^+$.

\subsection{Notation}
Let $T(V,E)$ be a finite, undirected tree with $n=\abs{V}$. For
$F\subseteq E$ let $\mathcal{C}_T(F)$ be the set of connected
components of $T\!\setminus\! F$. We define a tree $\mathcal{T}_F$ with
nodes ${\cal C}_T(F)$ and edges $F$. An edge of $(u,w)\in F$ connects
components $T_1,T_2 \in {\cal C}_T(F)$ if and only if $u\in T_1$ and
$w\in T_2$. For $F\subseteq E$ and $v\in V$ denote the number of edges
in $F$ incident to $v$ by $F_v$.

The nodes of a nontrivial tree $T$ can be uniquely partitioned into
two subsets, such that the nodes of each subset form an independent
set. In the following we denote these independent subsets by
$\mathcal{I}_0(T)$ and $\mathcal{I}_1(T)$. To enforce unambiguousness
when dealing with these subsets we demand that $v^\ast$ is contained
in $\mathcal{I}_0(T)$. A star graph is a tree with $n-1$ leaves. The
maximal degree of a tree is denoted by $\Delta$. We denote the
$n^{th}$ Fibonacci number by $F_n$, i.e., $F_0=0, F_1=1$, and
$F_n=F_{n-1}+F_{n-2}$. For a set $S$ we denote by $\mathcal{P}(S)$ the
power set of $S$, i.e., the set of all subsets of $S$.

\section{Fixed Points}  
In this section we provide a characterization of
$\mathcal{F}_{\mathcal{M}}(T)$ with respect to subsets of $E$. In
particular; we identify a set $E_{fix}(T)\subset \mathcal{P}(E)$ and
define a bijection $\mathcal{B}_{fix}$ between $E_{fix}(T)$ and
$\mathcal{F}_{\mathcal{M}}(T)^+$. $E_{fix}(T)\not= \emptyset$ since
$\emptyset \in E_{fix}(T)$. This shows that every tree has at least
one fixed point. The definition of $\mathcal{B}_{fix}$ is different
for $\mathcal{MIN}$ and $\mathcal{MAJ}$. These results allow to
characterize the fixed points of paths. In the second subsection we
prove an upper bound for $\abs{\mathcal{F}_{\mathcal{{M}}}(T)}$ in
terms of $n$ and $\Delta$. For the case of paths we give the exact
numbers. In the last part we provide an output-sensitive algorithm to
enumerate all fixed points.

\subsection{The Bijection $\mathcal{B}_{fix}$}\label{sec:fix}
For $c\in \mathcal{F}_{\mathcal{MIN}}(T)$ nodes adjacent to edges
monochromatic for $c$ have degree at least two, moreover at most one
half of the adjacent edges of each node are monochromatic for $c$.
Surprisingly the inverse of this statement is also true and forms the
basis for defining the bijection $\mathcal{B}_{fix}$: If $F$ is a
subset of the edges of $T$ such that nodes adjacent to edges in $F$
have degree at least two and at most one half of the adjacent edges of
each node are in $F$ then $F$ uniquely defines a fixed point of
$\mathcal{F}_{\mathcal{M}}(T)$.

\begin{lemma}\label{lem:basic}
  Let $T$ be a tree, $c\in \mathcal{F}_{\mathcal{M}}(T)$, and $F$ the
  set of monochromatic (resp.\ multicolored) edges $(u,w)\in E$ if
  $\mathcal{M}=\mathcal{MIN}$ (resp.\ $\mathcal{M}=\mathcal{MAJ}$). If
  $(u,w)\in F$ then $deg_T(u)\ge 2$ and $deg_T(w)\ge 2$. Furthermore,
  $F_v\le deg_T(v)/2$ for each node $v$ of $T$.
\end{lemma} 
\begin{proof}
  Assume $\mathcal{{M}}=\mathcal{{MIN}}$, the other case is proved
  similarly. Then $\abs{N_T^{1-c(u)}(u)} \ge \abs{N_T^{c(u)}(u)}\ge 1$
  for $(u,w)\in F$. Thus,
  $deg_T(u) = \abs{N_T^{c(u)}(u)} + \abs{N_T^{1-c(u)}(u)}\ge 2$.
  Similarly $deg_T(w)\ge 2$. Let $v\in V$. Then
  $\abs{N_T^{c(v)}(v)}\le\abs{N_T^{1-c(v)}(v)}$ since
  $c\in \mathcal{F}_{\mathcal{M}}(T)$, i.e.,
  $deg(v) \ge 2\abs{N_T^{c(v)}(v)}= 2F_v$.
\end{proof}

The last lemma motivates the following definition of $E_{fix}(T)$.
Note that $E_{fix}(T)$ satisfies the hereditary property
\begin{definition}
  Let $T$ be a tree. $E^2(T)$ denotes the set of edges of $T$ where
  each end node has degree at least two. $F \subseteq E^2(T)$ is
  called {\em legal} if $F_v \le deg(v)/2$ for each node $v$.
  $E_{fix}(T)$ denotes the set of all legal subsets of a tree $T$.
\end{definition}

\begin{theorem}\label{theo:min_fix}
  For any tree $T$ there exists a bijection $\mathcal{B}_{fix}$ between
  $E_{fix}(T)$ and $\mathcal{F}_{\mathcal{M}}(T)^+$.
\end{theorem}
\begin{proof}
  First assume $\mathcal{{M}}=\mathcal{{MIN}}$. Let
  $F\in E_{fix}(T)$. We define a coloring
  $c_F\in \mathcal{F}_{\mathcal{MIN}}(T)$. Let
  $T^\ast \in \mathcal{C}_T(F)$ with $v^\ast\in T^\ast$. Let
  $c_F(v^\ast)=0$ and extend $c_F$ to an independent coloring of
  $T^\ast$, e.g., by using breadth-first search. This uniquely defines
  $c_F$ on $ T^\ast$. We extend $c_F$ successively to a coloring with
  $c_F\in \mathcal{F}_{\mathcal{MIN}}(T)^+$. While there exists an
  already colored node $u$ that has an uncolored neighbor do the
  following. Let $T_1\in \mathcal{C}_T(F)$ with $u\in T_1$,
  $N_1=N_{T_1}(u)$, and $N_2=N_T(u)\setminus N_1$. All nodes in $N_1$
  have color $1-c_F(u)$ and $F_u=\abs{N_2}$. No node of $N_2$ has yet
  been assigned a color. By assumption we have
  $\abs{N_2} \le deg_T(u)/2$. Hence, $\abs{N_2} \le \abs{N_1}$. Set
  $c_F(w)=c_F(u)$ for all $w\in N_2$. For each $w\in N_2$ let
  $T_w\in \mathcal{C}_T(F)$ with $w\in T_w$. Extend $c_F$ to an
  independent coloring on each $T_w$. Then
  $\abs{N_T^{c_F(u)}}\le \abs{N_T^{1-c_F(u)}}$. Clearly this uniquely
  defines $c_F$ and $c_F \in \mathcal{F}_{\mathcal{MIN}}(T)^+$. Now we
  can define
  $\mathcal{B}_{fix}(F) = {c}_F \text{ for each } F\in E_{fix}(T)$.
  Let $F_1\not=F_2\in E_{fix}(T)$ and $e=(u,w)\in F_1\setminus F_2$.
  Then $c_{F_1}(w)=c_{F_1}(u)$ and
  $c_{F_2}(w)\not=c_{F_2}(u)$. Hence, $c_{F_1}\not= c_{F_2}$. Thus,
  $\mathcal{B}_{fix}(F)$ is injective.

  Next, we prove that $\mathcal{B}_{fix}$ is surjective, i.e., for
  every $c\in \mathcal{F}_{\mathcal{MIN}}(T)^+$ there exists
  $F_c\in E_{fix}(T)$ such that $\mathcal{B}_{fix}(F_c)=c$. For
  $c\in \mathcal{F}_{\mathcal{MIN}}(T)^+$ let
  $F_c=\{(u,w)\in E \suchthat c(u) = c(w)\}$. Then $F_c\in E_{fix}(T)$
  by Lemma~\ref{lem:basic}. By the first part of this proof we have
  $\mathcal{B}_T(F_c)\in \mathcal{F}_{\mathcal{MIN}}(T)^+$. Let
  $v\in T^\ast$ and $u \in N_{T^\ast}(v)$. Then $c(u)\not=c(v)$,
  otherwise $u\not\in T^\ast$. Hence, $\mathcal{B}_T(F_c)$ is for
  $T^\ast$ independent. Since $c_{F_c}(v^\ast) = c(v^\ast)=0$ we have
  $\mathcal{B}_T(F_c)(v) = c(v)$ for all $v\in T^\ast$. Next we repeat
  this argument for all $\hat{T}\in {\cal C}_T(F_c)$. Thus, $c$ and
  $\mathcal{B}_T(F_c)$ define the same coloring of $T$, i.e.,
  $\mathcal{B}_T(F_c)=c$.
  
  Next assume $\mathcal{{M}}=\mathcal{{MAJ}}$. Let $F\in E_{fix}(T)$.
  We use the partition of the nodes of $\mathcal{T}_F$ into two
  independent subsets $\mathcal{I}_0$ and $\mathcal{I}_1$ to define a
  mapping $C_F: {\cal C}_T(F) \rightarrow \{0,1\}$ by setting
  $C_F(\hat{T}) = i$ if $\hat{T}\in \mathcal{I}_i$ for $i=0,1$. We
  define a coloring $c_F$ of $T$ as follows
  \[c_F(v)=C_F(\hat{T}) \text{ if } v\in \hat{T}.\] $F$ uniquely
  defines $c_F$ among all colorings with $c_F(v^\ast)=0$, since
  $v^\ast \in T^\ast \in \mathcal{I}_0$. First, we prove that
  ${c}_F \in \mathcal{F}_{\mathcal{MAJ}}(T)^+$. For $v\in V$ let
  $\hat{T} \in {\cal C}_T(F)$ with $v \in \hat{T}$. Then
  $N(v)\cap \hat{T}=N_T^{c_F(v)}(v)$ and
  $F_u=\abs{N_T^{1-c_F(v)}(v)}$. Since $F\in E_{fix}(T)$ we have
  $\abs{N_T^{c_F(v)}(v)} \ge deg(v)/2$. Thus,
  $2\abs{N_T^{c_F(v)}(v)}\ge \abs{N_T^{c_F(v)}(v)}
  +\abs{N_T^{1-c_F(v)}(v)}$ and hence,
  $\abs{N_T^{c_F(v)}(v)}\ge\abs{N_T^{1-c_F(v)}(v)}$. Thus, no node can
  change its color. i.e., $c_F\in \mathcal{F}_{\mathcal{MAJ}}(T)^+$.
  Now we define
  \[\mathcal{B}_{fix}(F) = {c}_F \text{ for each } F\in E_{fix}(T).\]
  As in the proof for $\mathcal{{M}}=\mathcal{{MIN}}$ we can show that
  $\mathcal{B}_{fix}(F)$ is injective.

  Next, we prove that $\mathcal{B}_{fix}$ is surjective, i.e., for
  every $c\in \mathcal{F}_{\mathcal{MAJ}}(T)^+$ there exists
  $F_c\in E_{fix}(T)$ such that $\mathcal{B}_{fix}(F_c)=c$. For
  $c\in \mathcal{F}_{\mathcal{MAJ}}(T)^+$ define
  $F_c=\{(u,w)\in E \suchthat c(u) \not= c(w)\}$. Let $(u,w)\in F_c$.
  Then $F_c\in E_{fix}(T)$ by Lemma~\ref{lem:basic}. By the first part
  of this proof we have
  $\mathcal{B}_{fix}(F_c)\in \mathcal{F}_{\mathcal{MAJ}}(T)^+$.

  Let $v\in T^\ast$ and $u \in N_{T^\ast}(v)$. Then $c(u)=c(v)$,
  otherwise $u\not\in T^\ast$. Hence, $\mathcal{B}_{fix}(F_c)$ is for
  $T^\ast\in {\cal C}_T(F_c)$ monochromatic. Since
  $c_{F_c}(v^\ast) = c(v^\ast)=0$ we have
  $\mathcal{B}_{fix}(F_c)(v) = c(v)$ for all $v\in T^\ast$. Now we
  repeat this argument for all $\hat{T}\in {\cal C}_T(F_c)$. Thus, $c$
  and $\mathcal{B}_{fix}(F_c)$ define the same coloring, i.e.,
  $\mathcal{B}_{fix}(F_c)=c$.
\end{proof}

Theorem~\ref{theo:min_fix} implies the following two results.
\begin{corollary}\label{cor:fix_min_exits}
  Let $T$ be a tree. The minority process has an independent
  fixed point. It has a non-independent fixed point if and only if $T$
  has at least two inner nodes. The majority process has a
  monochromatic fixed point. It has a non-monochromatic fixed point if
  and only if $T$ has at least two inner nodes.
\end{corollary}
\begin{proof}
  Since $\emptyset \in E_{fix}(T)$ we have
  $c_\emptyset \in \mathcal{F}_{\mathcal{MIN}}(T)$ and $c_\emptyset$
  is an independent coloring. By Theorem~\ref{theo:min_fix} $T$ has a
  non-independent fixed point if and only if
  $E_{fix}(T)\not=\emptyset$. This is equivalent to having at least
  two inner nodes. The proof for $\mathcal{MAJ}$ is similar.
\end{proof}

\begin{corollary}\label{cor:fixed}
  A coloring of a path is a fixed point of the minority (resp.\
  majority) process if and only if each node has at least one
  neighbor with a different (resp.\ same)~color.
\end{corollary}
\begin{proof}
  Let $P_n$ be a path and $F\in E_{fix}(P_n)$. If $n\le 3$ then
  $E_{fix}(T)=\emptyset$, i.e., $\mathcal{F}_{\mathcal{MIN}}(P_n)$
  (resp.\ $\mathcal{F}_{\mathcal{MAJ}}(P_n)$) consist of the two
  independent (resp.\ monochromatic) colorings. Let $n\ge 4$. Then $F$
  is a matching of $P_{n-2}$ since the end edges of $P_n$ cannot be in
  $F$. Since $c_F$ induces an independent (resp.\ monochromatic)
  coloring for every $\hat{T}\in {\cal C}_T(F)$ and $\abs{\hat{T}}\ge
  2$ the proof is complete.
\end{proof}

\subsection{Counting Fixed Points}
\label{sec:count-fixed-points}

Theorem~\ref{theo:min_fix} allows to compute the number of fixed
points in specific cases. If $\Delta=n-1$ (resp.\ $\Delta=n-2$) then
$\abs{\mathcal{F}_{\mathcal{MIN}}(T)}= 2$ (resp.\
$\abs{\mathcal{F}_{\mathcal{MIN}}(T)}= 4$). Furthermore,
$\abs{\mathcal{F}_{\mathcal{MIN}}(T)}\le 8$ if $\Delta=n-3$. To get
more general results we describe an algorithm $\mathcal{A_{M}}$ to
generate all fix points of a given tree $T$. We start with node
$v^\ast$ and color it with $0$. Algorithm $\mathcal{A_{M}}$ is
recursive and extends a partial coloring by coloring all uncolored
neighbors of an already colored node. In this context a partial
coloring is a coloring of a subset of the nodes of $T$ with the
following property: Let $v$ be an already colored node. Firstly, all
nodes on the path from $v^\ast$ to $v$ in $T$ are also colored.
Secondly, if a neighbor of $v$ other than the one closer to $v^\ast$
is colored, then all neighbors of $v$ are colored.

The details of a recursive call for the minority process, i.e.,
$\mathcal{A_{MIN}}$ are as follows. Given a partial coloring $c$, a
single invocation generates several extensions of $c$, all of them are
again partial colorings covering more nodes. Let $v$ be an already
colored node that has an uncolored neighbor. First, each uncolored
neighbor of $v$ that is a leaf gets the complementary color of $v$.
Then $v$ has $r=deg(v) -\abs{N^{0}(v)} -\abs{N^{1}(v)}$ uncolored
neighbors. Let $U$ be the set of the uncolored neighbors of $v$, note
none of them is a leaf. We color $\hat{N}_0$ (resp.\ $\hat{N}_1$) of
these $r$ neighbors with color $0$ (resp.\ $1$), i.e.,
$r=\hat{N}_0 + \hat{N}_1$. In order to produce a fixed point the
following inequality must be satisfied:
\[ \abs{N^{c(v)}(v)} + \hat{N}_{c(v)} \le \abs{N^{1-c(v)}(v)} +
  \hat{N}_{1-c(v)}= \abs{N^{1-c(v)}(v)} + r - \hat{N}_{c(v)}.\] Hence,
\begin{equation}
  \label{eq:1}
\hat{N}_{c(v)} \le \frac{r + \abs{N_{1-c(v)}(v)} -\abs{N_{c(v)}(v)}}{2}.  
\end{equation}
Let
\begin{equation}
  \label{eq:2}
  r_0= \min \left(\lfloor (r +\abs{N_{1-c(v)}(v)}
    -\abs{N_{c(v)}(v)})/2\rfloor, r\right).
\end{equation}
For $i=0,\ldots, r_0$ we extend $c$ by coloring a subset $S$ of $U$ of
size $i$ with color $c(v)$ and the remaining nodes $U\setminus S$ with
color $c(v)-1$. This way we get $\sum_{i=0}^{r_0} \binom{r}{i}$
extended partial colorings. $\mathcal{A_{MIN}}$ is applied to each of
these extensions and terminates when all nodes are colored. Clearly,
the resulting colorings are fixed points and all fixed points are
generated this way. Algorithm $\mathcal{A_{MAJ}}$ differs only in two
places. Firstly, uncolored neighbors of $v$ that are leaves gets the
same color as $v$. Secondly, in Eq.~(\ref{eq:1}) operator $\ge$
must be replaced by $\le$ and the assignment of colors to nodes in $U$
is reversed.

Next we prove an upper bound for $\abs{{F}_{\mathcal{{M}}}(T)}$.
According to Corollary~\ref{cor:fix_min_exits} each tree has at least
two fixed points. A star graph is an extreme case, because it only has
two fixed points. The other extreme are paths as shown in this
section. We start with a simple observation.

\begin{lemma}\label{lem:4path}
  Let $T$ be a tree with a path $v_0,v_1,v_2,v_3$ such that
  $deg(v_0)=1$ and $deg(v_1)=deg(v_2)=2$. Let $T^0=T\setminus v_0$ and
  $T^1=T^0\setminus v_1$. Then
  $\abs{\mathcal{F}_{\mathcal{{M}}}(T)}=\abs{\mathcal{F}_{\mathcal{{M}}}(T^0)}+\abs{\mathcal{F}_{\mathcal{{M}}}(T^1)}$.
\end{lemma}
\begin{proof}
  We assume $\mathcal{{M}}=\mathcal{{MIN}}$, the other case is proved
  similarly. Note that for each $c\in \mathcal{F}_{\mathcal{{M}}}(T)$,
  the color of $v_0$ is determined by that of $v_1$, i.e.,
  $c(v_0)=1-c(v_1)$. We apply the above described algorithm
  $\mathcal{A_{\mathcal{MIN}}}$ to $T^1$ and $T^0$, where vertices
  $v_2,v_1$ and $v_0$ are visited last and in this order. We classify
  the fixed points $c$ of $T$ in two categories: $\mathcal{F}_=$
  consists of those $c$ with $c(v_2)=c(v_1)$ and $\mathcal{F}_{\not=}$
  the remaining colorings. Colorings of $\mathcal{F}_{\not=}$ arise
  when $\mathcal{A_{\mathcal{MIN}}}$ is applied to $T^0$, because
  $v_1$ is a leaf in $T^0$ and therefore $v_1$ and $v_2$ do receive
  different colors for $T^0$. Thus, each fixed point
  $c \in \mathcal{F}_{\not=}$ of this category can be uniquely paired
  with a fixed point of $T_0$, by setting $c(v_0) = 1- c(v_1)$.
  Consider an application of $\mathcal{A_{\mathcal{MIN}}}$ to $T^1$.
  Since $v_2$ is a leaf in $T^1$, $v_2$ and $v_3$ receive different
  colors. Thus, these colorings are uniquely extended to $T$ by
  setting $c(v_1) = c(v_2)$ and $c(v_0) = 1 - c(v_1)$. Hence, fixed
  points of $T^1$ can be uniquely paired with fixed points of
  $\mathcal{F}_=$. Since we have associated each fixed of $T$ with a
  fixed point of either $T^0$ or $T^1$ we have
  $\abs{\mathcal{F}_{\mathcal{MIN}}(T)}=\abs{\mathcal{F}_{\mathcal{MIN}}(T^0)}+\abs{\mathcal{F}_{\mathcal{MIN}}(T^1)}$. 
\end{proof}

Before proving Theorem~\ref{theo:count_fix} we provide a technical result.

\begin{lemma}\label{lem:rem}
  Let $T$ be a tree and $v$ a leaf of $T$ with neighbor $w$. Let $n_l$
  (resp.\ $n_{i}$) be the number of neighbors of $w$ that are leaves
  (resp.\ inner nodes). If $n_l>n_{i}$ then there is a one-to-one
  correspondence between $\mathcal{F}_{\mathcal{{M}}}(T)$ and
  $\mathcal{F}_{\mathcal{{M}}}(T\setminus v)$.
\end{lemma}
\begin{proof}
  For $c\in \mathcal{F}_{\mathcal{{M}}}(T)$ the children of $w$ that
  are leaves are colored with $c(w)$ if
  $\mathcal{{M}}=\mathcal{{MAJ}}$ and with $1-c(w)$ otherwise. Thus,
  if $c_1,c_2\in \mathcal{F}_{\mathcal{{M}}}(T)$ with $c_1(u)=c_2(u)$
  for all $u\not=v$ then $c_1=c_2$. Hence, $c$ uniquely corresponds to
  an element of $\mathcal{F}_{\mathcal{{M}}}(T\setminus v)$. Consider
  the coloring of $w$'s children by $\mathcal{A_{\mathcal{M}}}$. Since
  $n_l>n_{i}$ an arbitrary number of the non-leaf neighbors of $w$ can
  receive color $c(w)$ (resp.\ $1-c(w)$. If one of $w$'s leaf children
  is removed, $w$ still has at least as many leaf children as non-leaf
  children. Hence, $\mathcal{A_{\mathcal{M}}}$ produces the same fixed
  points.
\end{proof}

A path of length $3$ shows that the last lemma does not hold in case
$n_l=n_{i}$.

\begin{theorem}\label{theo:count_fix}
  Let $T$ be a tree and $P$ a path. Then
  $\abs{\mathcal{F}_{\mathcal{{M}}}(T)} \le
  2F_{n-\lceil\Delta/2\rceil}$ and
  $\abs{\mathcal{F}_{\mathcal{{M}}}(P)}= 2F_{n-1}$.
\end{theorem}
\begin{proof}
  We assume $\mathcal{{M}}=\mathcal{{MIN}}$, the other case is proved
  similarly. The proof is by induction on $n$. If $\Delta=2$ the
  result holds by Theorem~\ref{theo:count_fix}. If $T$ is a star graph
  then $\abs{\mathcal{F}_{\mathcal{{MIN}}}(T)}=2$, again the result is
  true. Let $\Delta> 2$ and $T$ not a star graph. Thus, $n>4$. Let
  $\tilde{E}$ be the set of edges $(v,w)$ of $T$ where $v$ is a leaf
  and all neighbors of $w$ but one are leaves. Then
  $\abs{\tilde{E}}\ge 2$ since $T$ is not a star graph. Thus, there
  exits $(v,w)\in \tilde{E}$ such that there exists a node different
  from $w$ that has degree $\Delta$. If $deg(w)>2$ then there exists a
  neighbor $u\not=v$ of $w$ that is a leaf. Let $T_u=T\setminus u$.
  Then
  $\abs{\mathcal{F}_{\mathcal{{M}}}(T)}=\abs{\mathcal{F}_{\mathcal{{M}}}(T_u)}$
  by Lemma~\ref{lem:rem}. Thus, the result is true by induction since
  $\Delta(T_u)=\Delta(T)$.
  

  Hence, we can assume that $deg(w)=2$. Let $u\not=v$ be the second
  neighbor of $w$. Denote by $T_v$ (resp.\ $T_w$) the tree
  $T\setminus v$ (resp.\ $T\setminus \{v,w\}$). Assume that $u$ is the
  only node in $T$ with degree $\Delta$. Repeating the above argument
  proves that $T$ is an extended star graph with center node $u$ and
  that all neighbors of $u$ have degree $1$ or $2$. Applying algorithm
  $\mathcal{A_{MIN}}$ shows that the largest number of fixed points is
  achieved if all $\Delta$ neighbors of $u$ have degree $2$. Then
  $n=2\Delta +1$, $r=\Delta$ and $r_0=\lfloor \Delta/2\rfloor$ by
  Eq.~(\ref{eq:2}). If $\Delta \equiv 0(2)$ then as shown above
  \[\abs{\mathcal{F}_{\mathcal{{M}}}(T)}=2\sum_{i=0}^{\Delta/2} \binom{\Delta}{i} =
    2^\Delta + \binom{\Delta}{\Delta/2}\le 2 F_{3\Delta/2 +1}=
    2F_{n-\lceil\Delta/2\rceil}.\] 
  The case $\Delta \equiv 1(2)$ is proved similarly. Therefore, we can
  assume that $T_w$ contains a node with degree $\Delta$. Thus, by
  induction:
  $2F_{n-1-\lceil\Delta/2\rceil} \ge
  \abs{\mathcal{F}_{\mathcal{{M}}}(T_v)}$ and
  $2F_{n-2-\lceil\Delta/2\rceil} \ge
  \abs{\mathcal{F}_{\mathcal{{M}}}(T_w)}$.

  Denote by $\mathcal{F}_=$ (resp.\ $\mathcal{F}_{\not=}$) the set of
  fixed points $c$ of $T$ with $c(u)= c(w)$ (resp.\ $c(u)\not= c(w)$).
  Note that $c(v) = 1- c(w)$. If $c\in \mathcal{F}_{\not=}$ then
  $c\in {\cal F}_{\mathcal{{M}}}(T_v)$ and if $c\in \mathcal{F}_=$
  then $c\in {\cal F}_{\mathcal{{M}}}(T_w)$. Hence,
  \[\abs{\mathcal{F}_{\mathcal{{M}}}(T)} = \abs{\mathcal{F}_{\not=}} +
  \abs{\mathcal{F}_=} \le 2F_{n-1-\lceil\Delta/2\rceil} +
  2F_{n-2-\lceil\Delta/2\rceil} = 2F_{n-\lceil\Delta/2\rceil}.\]
\end{proof}


Fig.~\ref{fig:noFib} shows that the bound of
Theorem~\ref{theo:count_fix} is not sharp. Let $B_h$ be a binary tree
of depth $h$. The equation
\[\abs{\mathcal{F}_{\mathcal{M}}(B_h)}
=\abs{\mathcal{F}_{\mathcal{M}}(B_{h-1})}(\abs{\mathcal{F}_{\mathcal{M}}(B_{h-1})}
+ 2\abs{\mathcal{F}_{\mathcal{M}}(B_{h-2})}^2)\] already contained in
\cite{Rastislav:2001} directly follows from
Theorem~\ref{theo:min_fix}.

\vspace*{-2mm}
\begin{figure}[h]
  \hfill
  \includegraphics[scale=0.95]{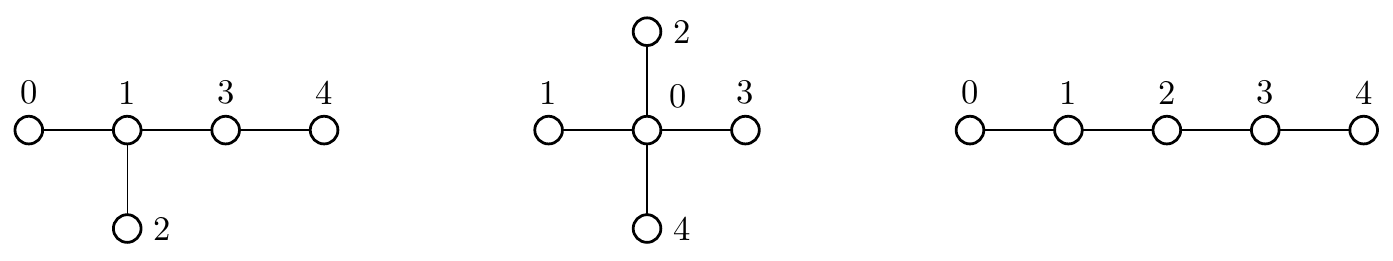}
  \hfill\null
  \caption{Three trees with five nodes having 4, 2, and 6 fixed points
    for $\mathcal{MIN}$.}\label{fig:noFib}
\end{figure}
\vspace*{-4mm}


\subsection{Generating Fixed Points}
\label{sec:enum-fixed-points}

The fixed points of a tree $T$ can be generated by iterating over all
subsets of $E^2(T)$ and outputting the legal ones. The algorithm exploits
the fact that $E_{fix}(T)$ has the hereditary property, i.e., if
$X\in E^2(T)$ is legal, all subset of $X$ are also legal.
Algorithm~\ref{alg:allfix} describes an output-sensitive algorithm
running in time
$O(n +\abs{\mathcal{F}_{\mathcal{M}}(T)}\times \abs{E^2(T)})$. Since
$\abs{E^2(T)} \le n$ the running time is in
$O(n\abs{\mathcal{F}_{\mathcal{M}}(T)})$. If
$E^2(T)=\{e_1,\ldots,e_l\}$ then the
algorithm successively constructs the set of all legal subsets using
the edges $\{e_1,\ldots,e_i\}$ for $i=0,\ldots,l$. The inner {\bf
  foreach}-loop always iterates over the list {\em fixedPoints}
beginning at the first entry.

 \SetProgSty{}
 \begin{algorithm}
  \Input{A tree $T=(V,E)$}
  \BlankLine
  $E^2 := \{(u,w)\in E\suchthat deg(u)\ge 2 \text{ and } deg(w)\ge
  2\}$\; $\mathrm{fixedPoints} :=\emptyset$;
  $\mathrm{fixedPoints.append}(\emptyset)$\;
  \ForEach{$e \in  E^2$}{
  $\mathrm{count}:=\mathrm{fixedPoints.size}()$\;
    \ForEach{$\mathrm{X}\in
      \mathrm{fixedPoints}$}{\If{$\{e\} \cup
        \mathrm{X} \text{ is
          legal}$}{$\mathrm{fixedPoints.append}(\mathrm{\{e\} \cup \mathrm{X}})$\;}
      $\mathrm{count}:=\mathrm{count}-1$\;
    \If{$\mathrm{count}==0$}{\Break\;}}
    }
  \KwRet{$\mathrm{fixedPoints}$\;}
  \caption{Algorithm to generate a list of all fixed points of a tree
    $T(V,E)$}\label{alg:allfix}
 \end{algorithm}

 \begin{theorem}
   Algorithm~\ref{alg:allfix} computes all $\abs{\mathcal{F}_{\mathcal{M}}(T)}$
   fixed points of a tree $T$ in time
   $O(n +\abs{\mathcal{F}_{\mathcal{M}}(T)}\times\abs{E^2(T)})$ using
   $O(\abs{E^2(T)}\times\abs{\mathcal{F}_{\mathcal{M}}(T)})$ memory.
 \end{theorem}
 \begin{proof}
   By Theorem~\ref{theo:min_fix} each legal subset of $E^2(T)$
   uniquely corresponds to a fixed point of $T$. If a subset $S$ of
   $E^2(T)$ is not legal, then no superset of $S$ is legal and if $S$
   is legal then all subsets of $S$ are legal. Therefore, the
   algorithm generates all legal subsets of $E^2(T)$. Let
   $l=\abs{E^2(T)}$. Denote by $S_i$ the set of elements of the list
   $fixedPoints$ at the beginning of the $i^{th}$ outer {\bf
     foreach}-loop and $S_{l+1}$ the elements of $fixedPoints$ after
   the last execution. Then $\abs{S_1}=1$ and
   $\abs{\mathcal{F}(T)^+}=\abs{S_{l+1}}$.

   Next we prove that $(4/5)\abs{S_{i+1}}\ge \abs{S_i}$ for
   $i=1,\ldots, l$. Let $e=(u,w)\in E^2(T)$. For $X\in S_i$ denote the
   number of edges in $X$ that are incident with a node $v$ by $X_v$.
   Let $\bar{S}=S_i$ and $\hat{S}=\emptyset$. Let $X\in \bar{S}$ with
   $X_u+1> deg(u)/2$ and $X_w+1> deg(w)/2$. Let $e_u$ (resp.\ $e_w$)
   be an edge of $X$ that is incident with $u$ (resp.\ $w$). Then we
   remove $X, X\setminus \{e_u,e_w\}, X\setminus \{e_u\}$, and
   $X\setminus \{e_w\}$ from $\bar{S}$ and insert
   $X, X\setminus \{e_u,e_w\}, X\setminus \{e_u\}$,
   $X\setminus \{e_w\}$, and $X\setminus \{e_u,e_w\}\cup \{e\}$ into
   $\hat{S}$. We repeat this process until there is no $X$ in
   $\bar{S}$ with the above property. Next, let $X\in \bar{S}$ with
   $X_u+1> deg(u)/2$ and $X_w+1\le deg(w)/2$. Let $e_u$ be an edge of
   $X$ that is incident with $u$. Then we remove $X$, and
   $X\setminus \{e_u\}$ from $\bar{S}$ and insert
   $X, X\setminus \{e_u\}, X\setminus \{e_u\}\cup \{e\}$ into
   $\hat{S}$. We repeat this process until there is no $X$ in
   $\bar{S}$ with the above property. Finally, for the remaining
   $X\in \bar{S}$ we insert $X, X\cup \{e\}$ into $\hat{S}$. Assume,
   that $S_i$ contains $n_1,n_2$ resp.\ $n_3$ elements according to
   the above classification, then \[\abs{S_i}=4n_1+2n_2+n_3 \text{ and }
   \abs{\hat{S}}=5n_1+3n_2+2n_3.\] Since $S_{i+1}= \hat{S}$ we have
   $(4/5)\abs{S_{i+1}}\ge \abs{S_i}$.
   The overall number of executions of the inner {\bf foreach}-loop is
   $\sum_{i=1}^l\abs{S_i}$. Thus,
   \[\sum_{i=1}^l\abs{S_i} \le   (4/5)\sum_{i=2}^{l+1}\abs{S_i}= 
     (4/5)\sum_{i=1}^{l}\abs{S_i} + (4/5)(\abs{S_{l+1}}-1).\]
   Hence,
   \[\sum_{i=1}^l\abs{S_i} \le 4 (\abs{S_{l+1}}-1) <
     4 \abs{\mathcal{F}(T)^+}.\] In time $O(n)$ we provide the degrees
   of all nodes in an array. Also the test whether $X\cup e$ is legal
   and append the entry to the list can be performed in time
   $O(\abs{X})$.
 \end{proof}

 The bound $(4/5)\abs{S_{i+1}}\ge \abs{S_i}$ for all $i$ can be used
 to prove the lower bound of $((5/4)^l$ with $l=\abs{E^2(T)}$ for
 $\abs{\mathcal{F}_{\mathcal{M}}(T)}$. We conjecture that a more detailed analysis
 of the relation between $\abs{S_{i+1}}$ and $\abs{S_{i}}$ leads to a
 better bound.

 Finally, we sketch an alternative approach for computing all fixed
 points. The elements of $E_{fix}(T)$ correspond to the solutions of a
 system of linear diophantine inequalities $Ax\le b$. Here, $A$ is a
 binary $\abs{E^2(T)} \times n$ matrix, where $a_{i,j}=1$ if node $i$
 is incident with edge $j$ of $E^2(T)$ and
 $b_i=\lfloor deg_T(i)/2\rfloor$. Thus, by Theorem~\ref{theo:min_fix}
 the set of fixed points corresponds to the solutions of $Ax\le b$.
 Unfortunately there isn't much work available for solving systems of
 linear diophantine inequalities \cite{Gao:2008}.


 \section{General 2-Cycles}
 In this section we analyze the structure of
 $\mathcal{C}^2_{\mathcal{M}}(T)$. First we collect general results
 about colorings from $\mathcal{C}^2_{\mathcal{MIN}}(T)$. In the
 second subsection we consider the set
 $c \in \mathcal{P}_{\mathcal{M}}(T)$ of all pure colorings. We first
 prove properties of $c$ and use these to define the set
 $E_{pure}(T)$ and define a bijection $\mathcal{B}_{pure}$ between
 $E_{pure}(T)$ and $\mathcal{P}_{\mathcal{M}}(T)^+$. Since
 $E_{pure}(T)\not= \emptyset$ this shows that every tree has pure
 coloring. These results immediately lead to a simple characterization
 pure coloring of paths. In the third subsection we derive from
 $\mathcal{B}_{pure}$ an upper bound for
 $\abs{\mathcal{P}_{\mathcal{{M}}}(T)}$ in terms of $n$. Finally we
 consider the general case of 2-cycles. We prove that $T$ decomposes
 into subtrees, such that $c$ is either a fixed point or a pure
 coloring on each of these subtrees.
 These subtrees provide the basis to define a hyper structure of a
 tree, called the {\em block tree}. After analyzing properties of block
 trees we define a set $E_{block}(T)$ of subsets of the edge set of a
 tree $T$ and show in Theorem~\ref{theo:main_2cycle_min} that the
 elements of $E_{block}(T)$ correspond one-to-one with the block trees
 of $T$. Since $E_{block}(T)$ does not have the hereditary property,
 we cannot use the approach of Algorithm~\ref{alg:allfix} to enumerate
 all block trees.

\subsection{General Results}
Let $c \in \mathcal{C}^2_{\mathcal{M}}(T)$. We separate the nodes of
$T$ in two groups. A node $u$ is called a {\em fixed node} for $c$ if
$\mathcal{M}(c)(u)=c(u)$; it is called a {\em toggle node} for $c$ if
$\mathcal{M}(c)(u)\not=c(u)$. Note that in any case
$\mathcal{M}(\mathcal{M}(c))(u)=c(u)$. Denote by $N^i_f(u)$ (resp.\
$N^i_t(u)$) the number of neighbors of $u$ with color $i$ that are
fixed (resp.\ toggle) nodes for $c$.

First, we provide a simple characterization of fixed and toggle nodes for
$\mathcal{MIN}$ and $\mathcal{MAJ}$.

\begin{lemma}\label{lem:basic_eq_min}
  Let $T$ be a tree and $c\in \mathcal{C}^2_{\mathcal{MIN}}(T)$. A
  node $u$ of $T$ is a fixed node of $c$ if and only if
  \[\abs{N_t^{1-c(u)}(u) - N_t^{c(u)}(u)} \le N_f^{1-c(u)}(u) -
  N_f^{c(u)}(u)\] and a toggle node of $c$ if and only if
  \[\abs{N_f^{c(u)}(u) - N_f^{1-c(u)}(u)} < N_t^{c(u)}(u) -
  N_t^{1-c(u)}(u).\]
\end{lemma}
\begin{proof}
  Let $u$ be a fixed node of $c$, i.e., $\mathcal{MIN}(c)(u)=c(u)$.
  Then
  $N_t^{c(u)}(u) + N_f^{c(u)}(u)\le N_t^{1-c(u)}(u) +
  N_f^{1-c(u)}(u)$. Since $\mathcal{M}(\mathcal{M}(c))(u)=c(u)$ we
  also have
  $N_t^{1-c(u)}(u) + N_f^{c(u)}(u)\le N_t^{c(u)}(u) +
  N_f^{1-c(u)}(u)$. This yields
  $-(N_f^{1-c(u)}(u) - N_f^{c(u)}(u)) \le N_t^{1-c(u)}(u) -
    N_t^{c(u)}(u) \le N_f^{1-c(u)}(u) - N_f^{c(u)}(u)$ which proves
  that the condition is necessary.

  Next assume
  $\abs{N_t^{1-c(u)}(u) - N_t^{c(u)}(u)} \le N_f^{1-c(u)}(u) -
  N_f^{c(u)}(u)$. Then
  $N_t^{c(u)}(u) - N_t^{1-c(u)}(u) \le N_f^{1-c(u)}(u) -
  N_f^{c(u)}(u)$ and $N^{c(u)}(u) \le N^{1-c(u)}(u)$. Hence,
  $\mathcal{M}(c)(u)=c(u)$. The assumption also implies that
  $N_t^{1-c(u)}(u) - N_t^{c(u)}(u) \le N_f^{1-c(u)}(u) -
  N_f^{c(u)}(u)$ resp.
  $N_f^{c(u)}(u) +N_t^{1-c(u)}(u) \le N_f^{1-c(u)}(u) +
  N_t^{c(u)}(u)$. This yields $\mathcal{M}(\mathcal{M}(c))(u)=c(u)$.
  Hence, the condition is sufficient. The result for toggle nodes is
  proved similarly.
\end{proof}

The proof of the following lemma is similar to the proof of
Lemma~\ref{lem:basic_eq_min}.

\begin{lemma}\label{lem:basic_eq_maj}
  Let $T$ be a tree and $c\in \mathcal{C}^2_{\mathcal{MAJ}}(T)$. A
  node $u$ of $T$ is a fixed node of $c$ if and only if
  \[\abs{N_t^{1-c(u)}(u) - N_t^{c(u)}(u)} \le N_f^{c(u)}(u) -
  N_f^{1-c(u)}(u)\] and a toggle node of $c$ if and only if
  \[\abs{N_f^{c(u)}(u) - N_f^{1-c(u)}(u)} < N_t^{1-c(u)}(u) -
  N_t^{c(u)}(u).\]
\end{lemma}

\subsection{Pure 2-Cycles}
If $c \in \mathcal{P}_{\mathcal{M}}(T)$ then each node of $T$ is a
toggle node. In Theorem~\ref{theo:main_pure} we give a
characterization $\mathcal{P}_{\mathcal{M}}(T)$, it allows to generate
all pure 2-cycles and compute $\abs{\mathcal{P}_{\mathcal{M}}(T)}$.

\begin{lemma}\label{lem:base:pure}
  Let $T$ be a tree, $c\in \mathcal{C}_{\mathcal{M}}(T)$. Then
  $c \in \mathcal{P}_{\mathcal{MIN}}(T)$ (resp.\
  $c \in \mathcal{P}_{\mathcal{MAJ}}(T)$) if and only if
  $N^{c(u)}(u)> N^{1-c(u)}(u)$ (resp.\ $N^{c(u)}(u)< N^{1-c(u)}(u)$)
  for each node $u$.
\end{lemma}
\begin{proof} We present the proof for $\mathcal{M}=\mathcal{MIN}$. If
  $N^{c(u)}(u)> N^{1-c(u)}(u)$ for all nodes $u$ then all nodes are
  toggle nodes. Hence, $N^{\mathcal{MIN}(c)(u)}(u)= N^{1-c(u)}(u)$.
  Thus, $N^{1-\mathcal{MIN}(c)(u)}(u) >N^{\mathcal{MIN}(c)(u)}(u)$,
  i.e., $\mathcal{MIN}(\mathcal{MIN}(c)) = c$, thus
  $c\in \mathcal{C}^2_{\mathcal{MIN}}(T)$. Since no nodes are fixed
  nodes we have $c\in \mathcal{P}_{\mathcal{MIN}}(T)$. The opposite
  statement follows from Lemma~\ref{lem:basic_eq_min}.
\end{proof}

As in Sec.~\ref{sec:fix} we use properties of monochromatic edges to
characterize pure 2-cycles. Corollary~\ref{col:monochr} is similar to
Lemma~\ref{lem:basic} and is used to define the set $E_{pure}(T)$.

\begin{lemma}\label{lem:Pure}
  Let $T$ be a tree, $c\in \mathcal{P}_{\mathcal{M}}(T)$, and
  $e=(u,w)\in E$ with $c(u)\not=c(w)$ if $\mathcal{M}=\mathcal{MIN}$
  and $c(u)=c(w)$ if $\mathcal{M}=\mathcal{MAJ}$. Let $T_u$ (resp.\
  $T_w$) be the subtree of $T\setminus e$ that contains $u$ (resp.\
  $w$). Then $u$ and $w$ have degree at least 3, $T_u$ and $T_w$
  contain at least 3 nodes, and $c$ induces a pure 2-cycle on both
  subtrees.
\end{lemma}
\begin{proof}
  We state the proof for $\mathcal{{M}}=\mathcal{{MIN}}$. Since $c$ is
  pure we have $N^{c(u)}_T(u) > N^{1-c(u)}_T(u)$ and since
  $c(u)\not=c(w)$ we also have $N^{1-c(u)}_T(u)\ge 1$. Hence,
  $deg(u)= N^{c(u)}_T(u) + N^{1-c(u)}_T(u)\ge 3$. Similarly
  $deg(w)\ge 3$.
  Let $v\in T_u$. If $v\not=u$ then all neighbors of $v$ in $T$ are in
  $T_u$ and thus
  $\abs{N^{c(u)}_{T_u}(u)} > \abs{N^{1-c(u)}_{T_u}(u)}$. Next consider
  the case $v=u$. Since $c$ is pure, there exists in $N(u)$ at least
  one more node with color $c(u)$ than with color $c(w)$. Thus, $u$
  has at least two neighbors in $T_u$, hence $T_u$ contains at least
  three nodes. Since
  $\abs{N^{c(u)}_{T_u}(u)}=\abs{N^{c(u)}_T(u)} > \abs{N^{1-c(u)}_T(u)}
  = \abs{N^{1-c(u)}_{T_u}(u)}+1$ we have
  $\abs{N^{c(u)}_{T_u}} > \abs{N^{1-c(u)}_{T_u}(u)}$. Hence,
  Lemma~\ref{lem:base:pure} implies that $c$ induces a pure 2-cycle
  for $\mathcal{{MIN}}$ on $T_u$. The same is true for $T_w$.
\end{proof}

\begin{corollary}\label{col:monochr}
  Let $T$ be a tree. If $c\in \mathcal{P}_{\mathcal{MIN}}(T)$,
  $F_c=\{(u,w)\in E \suchthat c(u) \not= c(w)\}$, and
  $\hat{T}\in {\cal C}_T(F_c)$ then $\abs{\hat{T}}\ge 3$ and $c$
  induces a monochromatic coloring on $\hat{T}$. If
  $c\in \mathcal{P}_{\mathcal{MAJ}}(T)$,
  $F_c=\{(u,w)\in E \suchthat c(u) = c(w)\}$, and
  $\hat{T}\in {\cal C}_T(F_c)$ then $\abs{\hat{T}}\ge 3$ and $c$
  induces an independent coloring on $\hat{T}$. Furthermore, $(F_c)_v
  < deg_T(v)/2$ for $v\in V$.
\end{corollary}
\begin{proof}
  Assume $\mathcal{{M}}=\mathcal{{MIN}}$. If
  $\hat{T}\in {\cal C}_T(F_c)$ then $c(u)=c(w)$ for all
  $u,w\in \hat{T}$. Thus, by Lemma~\ref{lem:Pure}, $c$ induces a
  monochromatic coloring on $\hat{T}$. Similarly, $c(u)\not=c(w)$ for
  all $u,w\in \hat{T}$ for $\mathcal{{M}}=\mathcal{{MAJ}}$. Again
  Lemma~\ref{lem:Pure}, $c$ induces an independent coloring on
  $\hat{T}$. $\abs{N^{c(v)}(v)}>\abs{N^{1-c(v)}(v)} = (F_c)_v$ since
  $c$ is pure. Hence, $deg_T(v)>2(F_c)_v$.
\end{proof}

Corollary~\ref{col:monochr} motivates the following definition of
$E_{pure}(T)$. Note that $E_{pure}(T)$ satisfies the hereditary
property and $E_{pure}(T)=E_{fix}(T)$ if all degrees of $T$ are odd.

\begin{definition}
  Let $T$ be a tree. $E^3(T)$ denotes the set of all edges of $T$
  where each end node has degree at least three. $F \subseteq E^3(T)$
  is called {\em legal} if $F_v < deg(v)/2$ for each node $v$.
  $E_{pure}(T)$ denotes the set of all legal subsets of $E^3(T)$.
\end{definition}

\begin{theorem}\label{theo:main_pure}
  For any tree $T$ there exists a bijection $\mathcal{B}_{pure}$ between
  $E_{pure}(T)$ and $\mathcal{P}_{\mathcal{M}}(T)^+$.
\end{theorem}
\begin{proof} 
  Let $F\in E_{pure}(T)$. We uniquely partition the nodes of
  $\mathcal{T}_F$ into two independent subsets $\mathcal{I}_0$ and
  $\mathcal{I}_1$ with $v^\ast\in \mathcal{I}_0$. Assume
  $\mathcal{{M}}=\mathcal{{MIN}}$. Define a mapping
  $C_F: {\cal C}_T(F) \rightarrow \{0,1\}$ by setting
  $C_F(\hat{T}) = i$ if $\hat{T}\in \mathcal{I}_i$. Based on $C_F$ we
  define a coloring $c_F$ of $T$ as follows
$c_F(v)=C_F(\hat{T}) \text{ if } v\in \hat{T}$. Note that
$c_F(v^\ast)=0$. $F$ uniquely defines $c_F$, since for each node $v$
there is a unique $\hat{T} \in {\cal C}_T(F)$ that contains $v$.
First, we prove that ${c}_F\in \mathcal{P}_{\mathcal{MIN}}(T)^+$. For
$v\in V$ let $\hat{T} \in {\cal C}_T(F)$ with $v \in \hat{T}$.
Then $N(v)\cap \hat{T}=N_T^{c_F(v)}(v)$. Since $F\in E_{pure}(T)$ we
have $\abs{N_T^{c_F(v)}(v)} > deg(v)/2$. Thus,
$2\abs{N_T^{c_F(v)}(v)}> \abs{N_T^{c_F(v)}(v)}
+\abs{N_T^{1-c_F(v)}(v)}$ and hence,
$\abs{N_T^{c_F(v)}(v)}>\abs{N_T^{1-c_F(v)}(v)}$ for all $v$. Hence,
$c_F\in \mathcal{P}_{\mathcal{MIN}}(T)^+$ by Lemma~\ref{lem:base:pure}.
Now we can define
$\mathcal{B}_{pure}(F) = {c}_F \text{ for each } F\in E_{pure}(T)$.
Let $F_1\not=F_2\in E_{pure}(T)$ and $e=(u,w)\in F_1\setminus F_2$.
Then $c_{F_1}(w)\not=c_{F_1}(u)$ and $c_{F_2}(w)=c_{F_2}(u)$. Hence,
$c_{F_1}\not= c_{F_2}$, i.e., $\mathcal{B}_{pure}(F)$ is injective.

Next, we prove that $\mathcal{B}_{pure}$ is surjective, i.e., for
every $c\in \mathcal{P}_{\mathcal{MIN}}(T)^+$ there exists
$F_c\in E_{pure}(T)$ with $\mathcal{B}_{pure}(F_c)=c$. For
$c\in \mathcal{P}_{\mathcal{MIN}}(T)^+$ define
$F_c=\{(u,w)\in E \suchthat c(u) \not= c(w)\}$. By
Lemma~\ref{lem:Pure} we have $F_c\in E^3(T)$. Let $v\in V$. Since $c$
is a pure 2-cycle we have $\abs{N_T^{c(v)}(v)}>\abs{N_T^{1-c(v)}(v)}$,
i.e., $deg(v) > 2\abs{N_T^{1-c(v)}(v)}$. Since,
$(F_c)_v =\abs{N_T^{1-c(v)}(v)}$ we have $deg(v)/2 > (F_c)_v$. This
yields $F_c\in E_{pure}(T)$. By the first part of this proof we have
$\mathcal{B}_{pure}(F_c)\in \mathcal{P}_{\mathcal{MIN}}(T)^+$. By
Corollary~\ref{col:monochr}  $\mathcal{B}_{pure}(F_c)$ is
for each tree $\hat{T}\in {\cal C}_T(F_c)$ a monochromatic coloring
with $\mathcal{B}_{pure}(F_c)(v) = c(v)$ for all $v\in \hat{T}$.
Hence, $c$ and $\mathcal{B}_{pure}(F_c)$ define the same coloring of
$T$, i.e., $\mathcal{B}_{pure}(F_c)=c$.

The proof for the case $\mathcal{{M}}=\mathcal{{MAJ}}$ is similar. The
main differences are that we define $c_F$ such that it induces an
independent coloring on each $\hat{T} \in {\cal C}_T(F)$ and in the
second part we define $F_c=\{(u,w)\in E \suchthat c(u) = c(w)\}$.
\end{proof}

\begin{corollary}\label{cor:pur_m_exits}
  Every tree $T$ has a pure coloring for the minority and the majority
  process. $T$ has a non-monochromatic (resp.\ non-independent) pure
  coloring for the minority (resp.\ majority) process if and only if
  there exist an edge $(u,w)\in T$ such that $deg(u)\ge3$ and
  $deg(w)\ge3$.
\end{corollary}
\begin{proof} We provide the proof for $\mathcal{M}=\mathcal{MIN}$.
  The result follows from Theorem~\ref{theo:main_pure}. Since
  $\emptyset \in E_{pure}(T)$ we have
  $c_\emptyset \in \mathcal{F}_{\mathcal{M}}(T)^+$. $c_\emptyset$ is a
  monochromatic coloring. $T$ has a non-monochromatic pure coloring if
  and only if $E_{pure}(T)\not=\emptyset$. This is equivalent to
  having an edge with the stated properties.
\end{proof}

\begin{corollary}\label{cor:2cyc}
  Let $P$ be a path and $ c\in \mathcal{C}(P)$. Then
  $c\in \mathcal{P}_{\mathcal{MIN}}(P)$ (resp.\
  $c\in \mathcal{P}_{\mathcal{MAJ}}(P)$) if and only if $c(v)= c(w)$
  (resp.\ $c(v)\not= c(w)$) for each edge $(v,w)$ of $P$.
\end{corollary}
\begin{proof} Clearly $E_{pure}=\emptyset$. Hence,
  $\mathcal{P}_{\mathcal{M}}(P)^+=\{c_\emptyset\}$ by
  Theorem~\ref{theo:main_pure}. Hence, $c$ is monochromatic (resp.\
  independent) for $\mathcal{M}=\mathcal{MIN}$ (resp.\
  $\mathcal{M}=\mathcal{MAJ}$).
\end{proof}

Since $E_{pure}(T)\subseteq E_{fix}(T)$ we have
$\mathcal{P}_{\mathcal{MAJ}}(T) \subseteq
\mathcal{F}_{\mathcal{MIN}}(T)$ and
$\mathcal{P}_{\mathcal{MIN}}(T) \subseteq
\mathcal{F}_{\mathcal{MAJ}}(T)$. Fig.~\ref{fig:outside} shows that
there are trees $T$ where
$\mathcal{P}_{\mathcal{MAJ}}(T) \subset
\mathcal{F}_{\mathcal{MIN}}(T)$ and
$\mathcal{P}_{\mathcal{MIN}}(T) \subset
\mathcal{F}_{\mathcal{MAJ}}(T)$. 

\begin{figure}[h]
  \hfill%
  \includegraphics[scale=0.95]{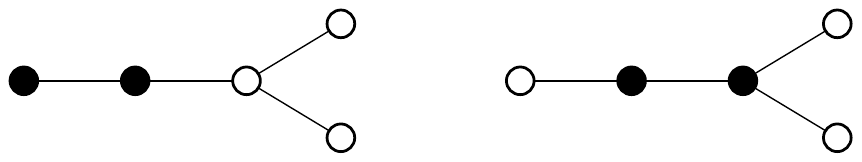}
  \hfill\null%
  \caption{The left coloring is in
    $\mathcal{F}_{\mathcal{MAJ}}(T)\setminus\mathcal{P}_{\mathcal{MIN}}(T)$,
    the right one is in
    $\mathcal{F}_{\mathcal{MIN}}(T)\setminus\mathcal{P}_{\mathcal{MAJ}}(T)$.}\label{fig:outside}
\end{figure}

\subsection{Counting Pure 2-Cycles}
Theorem~\ref{theo:main_pure} allows to determine the pure 2-cycles of
a tree $T$, and thus, $\abs{\mathcal{P}_{\mathcal{M}}(T)}$. Since
$E_{pure}(T) \subseteq E_{fix}(T)$ we have
$\abs{\mathcal{P}_{\mathcal{M}}(T)}\le
\abs{\mathcal{F}_{\mathcal{M}}(T)}$ and
$\abs{\mathcal{P}_{\mathcal{M}}(T)}\le2F_{n-\lceil\Delta/2\rceil}$ by
Theorem~\ref{theo:count_fix}. To generate all pure 2-cycles
Algorithm~\ref{alg:allfix} can be adopted, note that $E_{pure}(T)$ has
the hereditary property. The difference is that it uses $E^3(T)$ and
the corresponding notion of legal. The algorithm works in time
$O(n + \abs{\mathcal{P}_{\mathcal{M}}(T)}\abs{E^3(T)})$. Next we
provide a better upper bound for $\abs{\mathcal{P}_{\mathcal{M}}(T)}$.
Let $e_T=\abs{E^3(T)}$.

\begin{lemma}\label{lem:upper_bound}%
  Let $T$ be a tree, then $e_T\le (n-4)/2$.
\end{lemma}
\begin{proof}
  The proof is by induction on $n$. If $n<6$ then $e_T=0$ and for
  $n=6$ we have $e_T\le 1$. So let $n>6$. Let $v$ be a node with
  $deg(v)=2$. Let $T'$ be the tree that is constructed from $T$ by
  removing node $v$ and connecting the two neighbors of $v$ by an
  edge. Then $e_T=e_{T'}$. By induction $e_{T'} \le (n-5)/2< (n-4)/2$.  
  Hence we can assume that $deg(v)\not=2$ for all nodes $v$ of $T$.
  Let $v$ be a node with $deg(v)\ge 4$. Then by induction, no neighbor
  of a $v$ is a leaf. Hence, each neighbor of $v$ has degree at least
  $3$. Let $w$ be a neighbor of $v$ and $(v,w)$ an edge. Let $T_1$ be
  the connected component of $T\setminus e$ that contains $v$ and
  $T_2$ the other component with the additionally edge $e$.
  Let $T_i$ have $n_i$ nodes. Then $n_1+n_2=n+1$ and
  $e_T=1+e_{T_1} +e_{T_2}$. By induction $e_{T_i}\le (n_i-4)/2$. Thus,
  \[e_T\le 1 +(n_1-4)/2 + (n_2-4)/2 = (n_1+n_2 -6)/2 =
    (n-5)/2 < (n-4)/2.\] It remains the case each node has either degree $1$ or
  $3$. Let $l$ be the number of leaves of $T$, then there are $n-l$
  nodes of degree $3$. Then $e_T=n-l-1$ (remove all leaves, then
  $e_T$ edges remain). We have $l+3(e_T+1) = 2(n-1)$ and $l+e_T+1=n$.
  Thus $n-1-e_T +3(e_T+1)=2(n-1)$. This yields $2e_T= n-4$ which
  completes the proof.
\end{proof}

The last lemma implies
$\abs{\mathcal{P}_{\mathcal{M}}(T)}\le 2^{1+(n-4)/2}$. This bound is
purely based on the bound for $\abs{E^3(T)}$. By utilizing the
constraints imposed by $E_{pure}(T)$ better bounds may be derived. The
tree $H_n$ with $n\equiv 0(2)$ that consists of a path of length
$(n+2)/2$ and a single node attached to each inner node of the path
(see Fig.~\ref{fig:kamm}) shows that the bound of
Lemma~\ref{lem:upper_bound} is sharp, but there is large gap between
$\abs{E^3(H_n)}$ and $\abs{E_{pure}(H_n)}$.

\begin{figure}[h]
  \hfill%
  \includegraphics[scale=0.95]{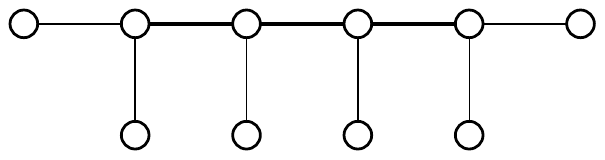}
  \hfill\null%
  \caption{The graph $H_{10}$, the three edges belonging to
    $E^3(H_{10})$ are depicted by solid lines. In general we have
    $\abs{E^3(H_n)}=2^{(n-4)/2}$ and $\abs{E_{pure}(H_n)}=F_{n/2}$.}\label{fig:kamm}
\end{figure}

\subsection{Block Trees of 2-Cycles}
In this section we consider general 2-cycles, i.e., those that have
both fixed and toggle nodes. We characterize the coarse grain
structure of $\mathcal{C}^2_{\mathcal{M}}(T)$, called the block tree
of $T$.

\begin{definition}
  Let $T$ be a tree and $c\in {\cal C}^2_{\mathcal{M}}(T)$. Let $V_f$
  (resp.\ $V_t$) be the set of fixed (resp.\ toggle) nodes of $c$ and
  $T^f$ (resp.\ $T^t$) the subgraph of $T$ induced by $V_f$ (resp.\
  $V_t$).
\end{definition}

The next result shows that a 2-cycle $c$ induces a structure on $T$ that
allows to define a hyper-tree $\mathcal{B}_c(T)$.

\begin{lemma}\label{lem:mainCycle}
  Let $T$ be a tree, $c\in {\cal C}^2_{\mathcal{M}}(T)$, and $T'$ a
  connected component of $T^f$ (resp.\ $T^t$). Then $c$ induces a
  fixed point (resp.\ a pure 2-cycle) on $T'$.
\end{lemma}
\begin{proof}
  We present the proof for $\mathcal{MIN}$. Let $T'$ be a connected
  component of $T^f$ and $u$ a node of $T'$. With respect to $T$ we
  have
  $\abs{N_t^{1-c(u)}(u) - N_t^{(u)}(u)} \le N_f^{1-c(u)}(u) -
    N_f^{c(u)}(u)$ by Lemma~\ref{lem:basic_eq_min}. Restricting $c$
  to $T'$ gives $N_{T'}^{c(u)}(u)=N_f^{c(u)}(u)$ and
  $N_{T'}^{1-c(u)}(u)=N_f^{1-c(u)}(u)$. This yields
  $N_{T'}^{1-c(u)}(u) \ge N_{T'}^{c(u)}(u)$. This proves that $u$ is a
  fixed node of $T'$ for $c$. Hence, $c$ is a fixed point for $T'$.
  The result about components of $T^t$ is proved similarly. The result for
  $\mathcal{MAJ}$ is based on Lemma~\ref{lem:basic_eq_maj}.
\end{proof}

Lemma~\ref{lem:mainCycle} provides the base to define the
{\em block tree} of a coloring $c\in {\cal C}^2_{\mathcal{M}}(T)$.
\begin{definition}
  Let $T$ be a tree, $c\in {\cal C}^2_{\mathcal{M}}(T)$, and
  $T_1,\ldots,T_s$ the connected components of $T^f$ and $T^t$. The
  block tree $\mathcal{B}_c(T)$ of $T$ for $c$ is a tree with nodes
  $\{T_1,\ldots,T_s\}$, nodes $T_i$ and $T_j$ are connected if there
  exists $(u,w)\in E$ with $u\in T_i$ and $w\in T_j$. A node $T_i$ is
  called a {\em fixed block} (resp. {\em toggle block}) of
  $\mathcal{B}_c(T)$ if $T_i$ is a connected component of $T^f$
  (resp.\ $T^t$).
\end{definition}

Obviously $\mathcal{B}_c(T)$ is a tree. $\mathcal{B}_c(T)$ is uniquely
defined, but different 2-cycles can induce the same block tree (see
Fig.~\ref{fig:Many_bct}). Each edge $e$ of $\mathcal{B}_c(T)$ connects
a fixed block with a toggle block, $e$ uniquely corresponds to an edge
of $T$. For convenience we denote this edge also by $e$. If $T_i$ is a
toggle block then obviously $\abs{T_i}\ge 2$, since all neighboring
blocks are fixed blocks. Fixed blocks can consist of a single node
only (see Fig.~\ref{fig:fixedSingle}).

The goal of this section is to present a characterization of the set
of all block trees for a given tree $T$ similar to
Theorem~\ref{theo:main_pure}, i.e., the trees $T_B$ for which there
exists $c\in {\cal C}^2_{\mathcal{M}}(T)$ such that
$T_B=\mathcal{B}_c(T)$. The following theorem summarizes properties of
2-cycles.

\begin{theorem}\label{thm:2cycmain}
  Let $T$ be a tree, $c\in \mathcal{C}^2_{\mathcal{M}}(T)$, and
  $e=(u,w)$ an edge of $\mathcal{B}_c(T)$. Then
  \begin{enumerate}
  \item If $deg_T(u)=2$ then $u$ is a fixed node. \label{lem:degree3}
  \item $\min(deg_T(u),deg_T(w))\ge 2$ and
    $\max(deg_T(u),deg_T(w))\ge 3$.\label{lem:fixedBlock}
  \item If $T_0$ is a node of $\mathcal{B}_c(T)$, $v\in T_0$,
    $deg_{T_0}(v)=1$ and $deg_T(v)\equiv 0(2)$ then $v$ is a fixed
    node and $T_0$ is a fixed block.\label{lem:fixedpt}
  \item If $T_0=\{v\}$ is a node of $\mathcal{B}_c(T)$ then $v$ is a
    fixed node, $T_0$ is a fixed block, and $deg_T(v)$ is
    even.\label{lem:singlept}
  \end{enumerate}
\end{theorem}
\begin{proof}
  Assume $\mathcal{M}=\mathcal{MIN}$, the proof for
  $\mathcal{MAJ}$ is similar and uses Lemma~\ref{lem:basic_eq_maj}.
  Assume that $u$ is toggle node. Then
  $\abs{N^{c(u)}(u)} > \abs{N^{1-c(u)}}$. Thus, if
  $\abs{N^{1-c(u)}}> 0$ then $deg_T(u)\ge 3$. Therefore,
  $\abs{N^{1-c(u)}}=0$ and $\abs{N^{c(u)}(u)}=2$. Since $u$ is toggle
  node, both neighbors must change their color, i.e., both are toggle
  nodes. This yields that $w$ is a toggle node. Contradiction, since
  $e(u,w)$ is an edge of $\mathcal{B}_c(T)$.
   
  WLOG we assume that $u$ is a fixed node while $w$ toggles its
  color. Assume that $\min(deg(u),deg(w))=1$. If $deg(u) =1$ then $u$
  cannot be a fixed node because $w$ toggles its color. Similarly,
  $w$ cannot have degree $1$. Hence, $\min(deg(u),deg(w))\ge 2$.
  Assume that $deg(u)=deg(w)=2$. Then by the first part, both nodes
  are fixed nodes. Contradiction.
  Assume that $v$ is a toggle node. Then $N_t^{c(v)}=1$ and
  $N_t^{1-c(v)}=0$. Hence, by Lemma~\ref{lem:basic_eq_min} we have
  $N_f^{1-c(v)}=N_f^{c(v)}$ thus,
  $deg_T(v)= 1 + 2 N_f^{c(v)}\equiv 1 (2)$. Contradiction.
  Let $T_0=\{v\}$. If $v$ is a toggle node then all neighbors are
  fixed nodes. Hence, $v$ is also a fixed node. Contradiction.
  Lemma~\ref{lem:basic_eq_min} yields that $deg_T(v)$ is even.
\end{proof}

\begin{figure}[h]
  \hfill
  \includegraphics[scale=0.95]{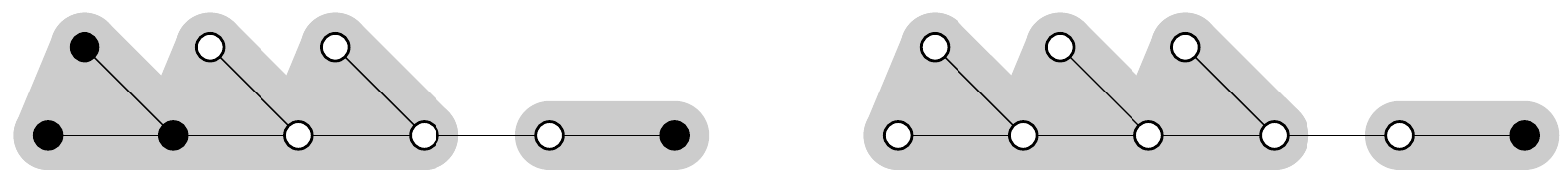}
  \hfill\null
  \caption{Two colorings leading to the same block tree. For the
    minority process both colorings define the same block tree. The
    left block node is a toggle node while the right is a fixed
    point.}\label{fig:Many_bct}
\end{figure}
\begin{figure}[h]
  \hfill%
  \includegraphics[scale=0.95]{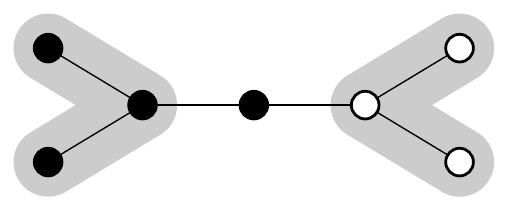}
  \hfill\null%
  \caption{A block tree consisting of two toggle blocks and one fixed
    block with a single node.}\label{fig:fixedSingle}
\end{figure}


The last theorem list properties of $\mathcal{B}_c(T)$ for
$c\in \mathcal{C}^2_{\mathcal{M}}(T)$. As before we take these
properties to identify a set of edges $F_c$ such that
${\mathcal{T}}_{F_c}= \mathcal{B}_c(T)$. The following two
definitions provide a formal framework for this purpose. 

\begin{definition}\label{def:block}
  Let $T$ be a tree. $E^{\,2.5}(T)$ denotes the set of edges of $T$,
  where one end node has degree at least two and the other has degree
  at least $3$. For $F\subseteq E^{\,2.5}(T)$ a component
  $\hat{T} \in {\cal C}_T(F)$ is called {\em fixed} if
  $\abs{\hat{T}}=1$ or if there exists $v\in \hat{T}$ such that
  $deg_T(v)\equiv 0(2)$ and $deg_{\hat{T}}(v)=1$. $Fix(T,F)$ denotes
  the set of all fixed components of ${\cal C}_T(F)$.
\end{definition}

\begin{definition}\label{def:legal3}
  Let $T$ be a tree.  $F \subseteq E^{\,2.5}(T)$ is called {\em legal} if
  all components of $Fix(T,F)$ are fully contained in
  $\mathcal{I}_0(\mathcal{T}_F)$ and if $T_0\in {\cal C}_T(F)$ with $T_0=\{v\}$
  then $deg_T(v)\equiv 0(2)$. $E_{block}(T)$ denotes the set of all
  legal subsets of $E^{\,2.5}(T)$.
\end{definition}

The next result reveals the significance of $E_{block}(T)$ for block
trees.

\begin{lemma}\label{lem:2cycle:req}
  Let $T$ be a tree, $c\in \mathcal{C}^2_{\mathcal{M}}(T)$, and $F_c$
  the edges of $\mathcal{B}_c(T)$. Then ${F_c}\in E_{block}(T)$.
\end{lemma}
\begin{proof}
  Note that $\mathcal{T}_{F_c}=\mathcal{B}_c(T)$. By
  Theorem~\ref{thm:2cycmain}.\ref{lem:fixedBlock} we have
  $F_c \subseteq E^{\,2.5}(T)$. By construction of $\mathcal{B}_c(T)$
  and Theorem~\ref{thm:2cycmain}.\ref{lem:singlept} and
  \ref{thm:2cycmain}.\ref{lem:fixedpt} we have
  $Fix(T,F_c)\subseteq \mathcal{I}_0(\mathcal{T}_{F_c})$.
  Theorem~\ref{thm:2cycmain}.\ref{lem:singlept} completes the proof.
\end{proof}

\begin{definition}
  Let $T$ be a tree. A coloring $c\in {\cal C}^2_{\mathcal{MIN}}(T)$
  is called {\em canonical} if $c$ induces a monochromatic (resp.\
  independent) coloring on each connected component of $T^t$ (resp.\
  $T^f$). A coloring $c\in {\cal C}^2_{\mathcal{MAJ}}(T)$ is called
  {\em canonical} if $c$ induces an independent (resp.\ monochromatic)
  coloring on each connected component of $T^t$ (resp.\ $T^f$).
\end{definition}

The next result lays the groundwork for our characterization of block
trees.

\begin{lemma}\label{lem:2cycle:suf}
  Let $T$ be a tree and $F\in E_{block}(T)$. There exits
  $c\in \mathcal{C}^2_{\mathcal{M}}(T)$ with
  $\mathcal{B}_c(T) = \mathcal{T}_F$ such that $c$ is canonical and
  $\mathcal{I}_0$ (resp.\ $\mathcal{I}_1$) is the set of fixed (resp.\
  toggle) nodes of $c$.
\end{lemma}

\begin{proof}
  Assume $\mathcal{M}=\mathcal{MIN}$, $\mathcal{M}=\mathcal{MAJ}$ is
  similar and uses Lemma~\ref{lem:basic_eq_maj}. The proof is by
  induction on $\abs{F}$. The case $\abs{F}=0$ is obvious, $c$ is the
  monochromatic coloring. Let $\abs{F}>0$. Let $L\in \mathcal{T}_{F}$
  be a leaf and $e=(u,w)\in F$ such that $w\in L$. Then $\abs{L}\ge 2$
  if $L\in \mathcal{I}_0(\mathcal{T}_F)$ and $\abs{L}\ge 3$ if
  $L\in \mathcal{I}_1(\mathcal{T}_F)$. Remember that
  $\mathcal{I}_0(\mathcal{T}_F)$ contains the fixed components of
  $\mathcal{T}_F$. By the definition of $ E_{block}(T)$ we have to
  consider four cases.
  
  {\bf Case 1:} $L\in
  \mathcal{I}_0(\mathcal{T}_F)$ and $\abs{L} > 2$.
  We construct a tree $\tilde{T}$ as follows: Remove from $T$ all
  nodes of $L$ except $w$ and add a new neighbor $v$ to $w$. Then
  $\abs{\tilde{T}} < \abs{T}$. Then $deg_T(u)>2$ otherwise $L$ would
  not be in $\mathcal{I}_0(\mathcal{T}_F)$. Hence,
  $F\subseteq E^{2.5}(\tilde{T})$. Denote the leaf of
  ${\cal C}_{\tilde{T}}(F)$ consisting of $v$ and $w$ by $\tilde{L}$.
  Thus, $\tilde{L}\in Fix(\tilde{T},F)$ and
  $Fix(\tilde{T},F) = Fix(T,F) \cup \tilde{L} \setminus L \subseteq
  \mathcal{I}_0(\mathcal{T}_F)$. Let $T_0=\{v\} \in {\cal C}_{\tilde{T}}(F)$.
  Then, $T_0\in {\cal C}_{{T}}(F)$. Hence, $deg_T(v)\equiv 0(2)$ by
  assumption. Since $T_0\in \mathcal{I}_0(\mathcal{T}_F)$ we also have
  $deg_{\tilde{T}}(v)\equiv 0(2)$. This shows that $\tilde{T}$ and $F$
  satisfy the theorem's assumption. Hence, by induction there exists a
  canonical coloring $\tilde{c}\in {\cal C}^2(\tilde{T})$ with
  $\mathcal{B}_{\tilde{c}}(\tilde{T}) = \mathcal{T}_F$ satisfying all
  properties. We can extend $\tilde{c}$ to a coloring
  $c\in {\cal C}^2({T})$ by setting $c(x)=\tilde{c}(x)$ for all nodes
  $x\in T\setminus L$, $c(w) = \tilde{c}(w)$, and color the remaining
  nodes of $L$ in the canonical way for a fixed point.

    {\bf Case 2:} $L\in
  \mathcal{I}_0(\mathcal{T}_F)$ and $\abs{L} = 2$.
  Let $\tilde{F}= F\setminus e$. Let $v\in L$ be a neighbor of $w$ and
  set $\tilde{T}= T\setminus v$. Let $T_u \in {\cal C}_{{T}}({F})$
  with $u\in T_u$. Then $T_u \in \mathcal{I}_1(\mathcal{T}_F)$ and thus,
  $\abs{T_u}>1$, $deg_{T_u}(u)\ge 1$ and $deg_{T}(u)\ge 3$. Let
  $\tilde{T}_u \in {\cal C}_{\tilde{T}}(\tilde{F})$ with
  $u\in \tilde{T}_u$. Then $w \in \tilde{T}_u$,
  $\tilde{T}_u \in \mathcal{I}_1(\mathcal{T}_F)$ and $T_u \subset \tilde{T}_u$.
  Clearly, $\tilde{F} \subseteq E^{2.5}(\tilde{T})$. Let
  $T_0 =\{v_0\} \in {\cal C}_{\tilde{T}}(\tilde{F})$ with
  $\abs{T_0}=1$. Then $T_0 \in {\cal C}_{{T}}({F})$, thus
  $deg_T(v_0)\equiv 0(2)$. Hence, $deg_{\tilde{T}}(v_0)\equiv 0(2)$.
  Let $\hat{T}\in {\cal C}_{\tilde{T}}(\tilde{F})$ and
  $v_0\in \hat{T}$ with $deg_{\hat{T}}(v_0)=1$,
  $deg_{\tilde{T}}(v_0)\equiv 0(2)$. Assume~$\hat{T} = \tilde{T}_u$.
  Then $v_0\not= w$ since $deg_{\tilde{T}}(w)=1\not\equiv 0(2)$. Thus,
  $\hat{T} = \tilde{T}_u$ if $v_0\in \hat{T}$ with
  $deg_{\hat{T}}(v_0)=1$ for some $v_0\not=w$. Hence,
  $\hat{T}\in Fix(\tilde{T},\tilde{F}) = Fix(T,F)\subseteq
  \mathcal{I}_0(\mathcal{T}_F)=
 \mathcal{I}_0(\mathcal{\tilde{T}}_{\tilde{F}})$.

  Therefore, $\tilde{T}$ and $\tilde{F}$ satisfy the theorem's
  assumption. By induction there exists a canonical coloring
  $\tilde{c}\in {\cal C}^2(\tilde{T})$ with
  $\mathcal{B}_{\tilde{c}}(\tilde{T}) = \mathcal{T}_{\tilde{F}}$
  satisfying all properties. Since $\tilde{T}_u\in \mathcal{I}_1(T)$
  all nodes of $\tilde{T}_u$ have the same color, thus
  $N_t^{1-\tilde{c}(u)}(u)=0$ and $\tilde{c}(u)=\tilde{c}(w)$. By
  Lemma~\ref{lem:basic_eq_min} we have
  $\abs{N_f^{\tilde{c}(u)}(u) - N_f^{1-\tilde{c}(u)}(u)} <
  N_t^{\tilde{c}(u)}(u)$.
  We change $\tilde{c}$ to a coloring $c$ of $T$ as follows. First, we
  set $c(x)= \tilde{c}(x)$ for all $x \not\in \{w,v\}$. We apply
  Lemma~\ref{lem:basic_eq_min} to prove that $u$ is still a toggle
  node for $c$.
  
  If $N_f^{\tilde{c}(u)}(u) > N_f^{1-\tilde{c}(u)}(u)$ we set
  $c(w)=1-\tilde{c}(w)$ and $c(v)=\tilde{c}(w)$. If
  $N_f^{\tilde{c}(u)}(u) < N_f^{1-\tilde{c}(u)}(u)$ we set
  $c(w)=\tilde{c}(w)$ and $c(v)=1-\tilde{c}(w)$. At last consider the
  case $N_f^{\tilde{c}(u)}(u) = N_f^{1-\tilde{c}(u)}(u)$. If
  $N_t^{\tilde{c}(u)}(u)=2$ then $N_t^{{c}(u)}(u)=1$, i.e.,
  $deg_{{T_u}}(u)=1$. Also
  $deg_{\tilde{T}}(u) = 2N_f^{\tilde{c}(u)}(u) + 2$, i.e,
  $deg_{{T}}(u) \equiv 0(2)$. Hence,
  $T_u\in \mathcal{I}_0(\mathcal{T}_F)$. Contradiction and thus
  $N_t^{\tilde{c}(u)}(u)>2$. Set $c(w)=1-\tilde{c}(w)$ and
  $c(v)=\tilde{c}(w)$. Then $N_t^{{c}(u)}(u)>1$ and thus,
  $\abs{N_f^{c(u)} - N_f^{1-c(u)}} =1 < N_t^{{c}(u)}(u)$. Therefore,
  $c$ has the desired properties.

  {\bf Case 3:} $L\in
  \mathcal{I}_1(T)$ and $\abs{L} > 3$.\\
  Construct a tree $\tilde{T}$ as follows: Remove from $T$ all nodes
  of $L$ except $w$ and add two new neighbors $v_1,v_2$ to $w$. Then
  $\abs{\tilde{T}} < \abs{T}$. Note that $deg_T(w)\ge 3$ otherwise
  $L \in \mathcal{I}_0(\mathcal{T}_F)$. Hence, $F\subseteq E^{2.5}(\tilde{T})$.
  Denote the leaf of $\mathcal{I}_{\tilde{T}}(F)$ consisting of
  $v_1,v_2$ and $w$ by $\tilde{L}$. Then
  $\tilde{L}\in \mathcal{I}_1(T)$ and hence,
  $Fix(\tilde{T},F) = Fix(T,F) \subseteq \mathcal{I}_0(\mathcal{T}_F)$. Let
  $T_0=\{v\} \in {\cal C}_{\tilde{T}}(F)$. Then,
  $T_0\in {\cal C}_{{T}}(F)$. Hence, $deg_T(v)\equiv 0(2)$ by
  assumption. Then also $deg_{\tilde{T}}(v)\equiv 0(2)$. This shows
  that $\tilde{T}$ and $F$ satisfy the assumption of the theorem.
  Hence, by induction there exists a canonical coloring
  $\tilde{c}\in {\cal C}^2(\tilde{T})$ with
  $\mathcal{B}_{\tilde{c}}(\tilde{T}) = \mathcal{T}_F$. We extend
  $\tilde{c}$ to a coloring $c\in {\cal C}^2({T})$ by setting
  $c(x) = \tilde{c}(w)$ for all $x\in L$ and $c(x)=\tilde{c}(x)$ for
  all other nodes $x$. Clearly $c$ satisfies the required conditions.

  {\bf Case 4:} $L\in
  \mathcal{I}_1(T)$ and $\abs{L} = 3$.\\
  Let $\tilde{F}= F\setminus e$. Since $L\in \mathcal{I}_1(T)$ we have
  $deg_{L}(w)=2$. Let $v_1,v_2$ be the neighbors of $w$ in $L$ and set
  $\tilde{T}= T\setminus v_1$. Let $T_u \in {\cal C}_{{T}}({F})$ with
  $u\in T_u$. Then $T_u \in \mathcal{I}_0(\mathcal{T}_F)$. Let
  $\tilde{T}_u \in {\cal C}_{\tilde{T}}(\tilde{F})$ with
  $u\in \tilde{T}_u$. Then $\tilde{T}_u \in \mathcal{I}_0(\mathcal{T}_F)$ and
  $T_u \subseteq \tilde{T}_u$.

  Clearly, $\tilde{F} \subseteq E^{2.5}(\tilde{T})$.
  Let $T_0 =\{v_0\} \in {\cal C}_{\tilde{T}}(\tilde{F})$ with
  $\abs{T_0}=1$. Then $T_0 \in {\cal C}_{{T}}({F})$, thus
  $deg_T(v_0)\equiv 0(2)$. Hence, $deg_{\tilde{T}}(v_0)\equiv 0(2)$.

  Let $\hat{T}\in {\cal C}_{\tilde{T}}(\tilde{F})$ and
  $v_0\in \hat{T}$ with $deg_{\hat{T}}(v_0)=1$ and
  $deg_{\hat{T}}(v_0)\equiv 0(2)$. If $\hat{T} = \tilde{T}_u$ then
  $\hat{T}\in \mathcal{I}_0(\mathcal{T}_F)$. Otherwise
  $\hat{T} \in {\cal C}_{{T}}({F})$ and thus,
  $\hat{T}\in \mathcal{I}_0(\mathcal{T}_F)$. Therefore, $\tilde{T},\tilde{F}$
  satisfy the assumption of the theorem. By induction there exists a
  canonical coloring $\tilde{c}\in {\cal C}^2(\tilde{T})$ with
  $\mathcal{B}_{\tilde{c}}(\tilde{T}) = \mathcal{T}_{\tilde{F}}$.
  Since $\tilde{T}_u\in \mathcal{I}_0(\tilde{T}_{\tilde{F}})$ neighboring nodes of
  $\tilde{T}_u$ have different colors. Since
  $\tilde{T}_u\in \mathcal{I}_0(\mathcal{T}_F)$ neighboring nodes of $\tilde{T}_u$
  have different colors, thus $N_f^{c(u)}(u)=0$ and
  $N_f^{1-c(u)}(u)>0$. By Lemma~\ref{lem:basic_eq_min} we have
  $\abs{N_t^{1-c(u)}(u) - N_t^{c(u)}(u)} < N_f^{1-c(u)}(u)$. We change
  $\tilde{c}$ to a coloring $c$ of $T$ as follows. First, we set
  $c(x)= \tilde{c}(x)$ for all $x \not\in \{w,v_1,v_2\}$. If
  $N_t^{1-c(u)}(u) > N_t^{c(u)}(u)$ we set $c(w)=\tilde{c}(u)$. If
  $N_t^{1-c(u)}(u) < N_t^{c(u)}(u)$ we set $c(w)=1-\tilde{c}(u)$. In
  both cases we set $c(v_1)=c(v_2)=c(w)$. At last consider the case
  $N_t^{c(u)}(u) = N_t^{1-c(u)}(u)$. Since $N_f^{1-c(u)}(u)>0$ we can
  take any of the two approaches. Clearly $c$ has the desired
  properties.
\end{proof}

\begin{theorem}\label{theo:main_2cycle_min}
  For a tree $T$ there exists a bijection $\mathcal{B}_{block}$
  between $E_{block}(T)$ and the set of block trees of $T$ of the
  minority and the majority process.
\end{theorem}
\begin{proof}
  The existence of $\mathcal{B}_{block}$ follows from
  Lemma~\ref{lem:2cycle:req} and \ref{lem:2cycle:suf}.
\end{proof}

The following result is an immediate implication of the last theorem.
\begin{corollary}
  Let $T$ be a tree where all nodes have odd degree. Then
  \[E_{block}(T) = \{ F\subseteq E^{3} (T) \suchthat {\cal C}_T(F)
    \text{ does not contain a component of size } 1\}.\] If $P$ is a
  path then
  $\mathcal{C}^2_{\mathcal{M}}(P)=\mathcal{P}_{\mathcal{M}}(P)$.
\end{corollary}


\subsection{Counting Block Trees}

The concept of Algorithm~\ref{alg:allfix} can not be used to generate
all elements of $E_{block}(T)$ because $E_{block}(T)$ does not have
the hereditary property (see Fig.~\ref{fig:gegenbeispiel}). Since
$E_{block}(T) \subseteq E_{fix}(T)$ each upper bound for
$\abs{\mathcal{F}_{\mathcal{M}}(T)}$ is also an upper bound for
$\abs{\mathcal{C}^2_{\mathcal{M}}(T)}$. A naive way to generate all
block trees of a tree is to iterate over the set $E_{fix}(T)$ and
test, whether an element is legal according to Def.~\ref{def:legal3}.

\begin{figure}[h]
  \hfill
  \includegraphics[scale=0.95]{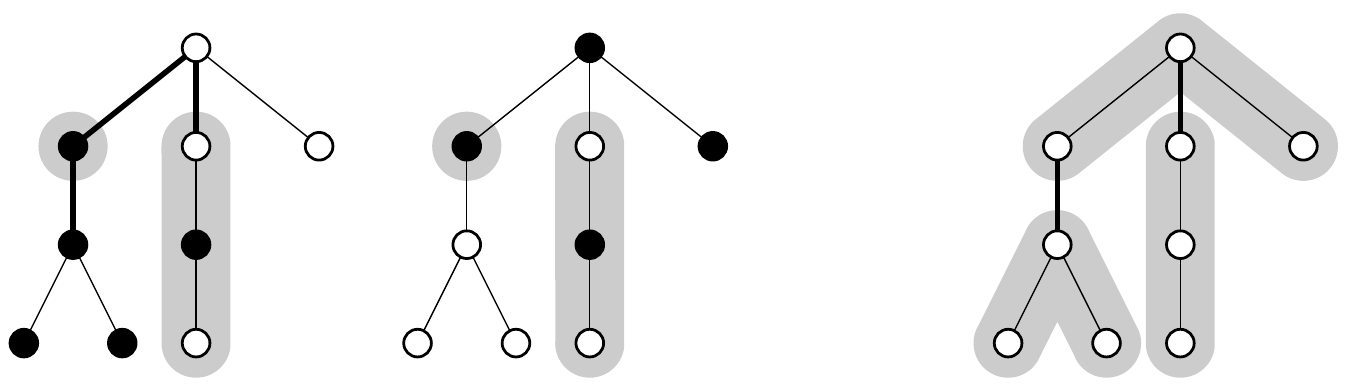}
  \hfill\null

  \caption{The left two drawings show a coloring from
    $\mathcal{C}^2_{\mathcal{M}}(T)$ which corresponds to the set
    $F\in E_{block}(T)$ which consists of three bold edges. The two
    fixed blocks are highlighted. In the right most drawing the tree
    $\mathcal{T}_{F'}$ for the set $F'\subset F$ of the two marked edges is
    shown. $F'$ is not legal, because two neighboring components are
    fixed by Def.~\ref{def:block}.}\label{fig:gegenbeispiel}
\end{figure}

\section{Conclusion and Open Problems}
In this paper we provided characterizations of several categories of
colorings of trees for the minority and majority process in terms of
subsets of the tree edges. This means that the class of trees is the
first nontrivial graph class for which a complete characterization of
fixed points for the minority/majority process exists. This includes
an algorithm to enumerate all fixed points and upper bounds for the
number of fixed points.

There are several open questions that are worth pursuing. Firstly, is
it possible to characterize fixed points and pure colorings for other
graph classes? Clearly, the results for trees do not hold for general
graphs, e.g.\ for cycles. But, it might be possible to use the same
approach, i.e., find suitable subsets of the edge set similar to
$E_{fix}$.

Furthermore, the current work for trees can be improved. It would be
interesting to find better general upper bounds for
$\abs{\mathcal{F}_{\mathcal{M}}(T)}$ and
$\abs{\mathcal{P}_{\mathcal{M}}(T)}$ for trees. Also, we believe that
the run-time of Algorithm~\ref{alg:allfix} can be improved. Moreover,
an algorithm to enumerate all block trees is an open problem. Finally,
a full characterization of all 2-cycles with the help of a subset of
the power set of the tree edges is still missing.

Another line of research is to consider random trees and compute the
expected number of fixed points and pure colorings. Using our results,
it suffices to compute the expected sizes of $\abs{E_{fix}(T)}$ and
$\abs{E_{pure}(T)}$ for these trees.

\bibliographystyle{plainurl}
\bibliography{arxiv}
\end{document}